\documentclass[a4paper,11pt]{article}

\usepackage{microtype}
\usepackage{amssymb,amsmath,amsthm} \usepackage{graphicx}
\usepackage{xspace} \usepackage{xypic} \usepackage[numbers]{natbib}
\usepackage{hyperref,color}
\usepackage[T1]{fontenc}
\usepackage[mathlines]{lineno} 
\usepackage{fullpage}
\usepackage{pdfpages}

\newcommand{\subp}{\ensuremath{Q}}
 
\newtheorem{theorem}{Theorem} 
\newtheorem{lemma}[theorem]{Lemma}

\newtheorem{corollary}[theorem]{Corollary}
\newtheorem{definition}[theorem]{Definition}

\newbox\ProofSym \setbox\ProofSym=\hbox{%
	\unitlength=0.18ex%
	\begin{picture}(10,10) \put(0,0){\framebox(9,9){}}
	\put(0,3){\framebox(6,6){}}
	\end{picture}}

\title{A New Balanced Subdivision of a Simple Polygon for Time-Space
  Trade-off Algorithms%
  \thanks{This research was supported by the MSIT(Ministry of Science
    and ICT), Korea, under the SW Starlab support
    program(IITP-2017-0-00905) supervised by the IITP(Institute for
    Information \& communications Technology Promotion)}}

\author{Eunjin Oh\thanks{Max Planck Institute for Informatics, Saarbr\"ucken, Germany, Email: \texttt{eoh@mpi-inf.mpg.de}} \and Hee-Kap Ahn\thanks{Department of Computer Science and
		Engineering, Pohang, POSTECH, Korea, Email: \texttt{heekap@postech.ac.kr}}}

\newcommand{\bd}{\ensuremath{\partial}}
\newcommand{\lcc}{\ensuremath{($\textsc{l}$,$\textsc{cc}$)}}
\newcommand{\lc}{\ensuremath{($\textsc{l}$,$\textsc{c}$)}}
\newcommand{\rcc}{\ensuremath{($\textsc{r}$,$\textsc{cc}$)}}
\newcommand{\rc}{\ensuremath{($\textsc{r}$,$\textsc{c}$)}}

\begin{document}
\date{}
\maketitle

\begin{abstract}
  We are given a read-only memory for input and a write-only stream
  for output.  For a positive integer parameter $s$, an $s$-workspace
  algorithm is an algorithm using only $O(s)$ words of workspace in
  addition to the memory for input.  In this paper, we present an
  $O(n^2/s)$-time $s$-workspace algorithm for subdividing a simple
  polygon into $O(\min\{n/s,s\})$ subpolygons of complexity
  $O(\max\{n/s,s\})$.	
  As applications of the subdivision, the previously best known
  time-space trade-offs for the following three geometric problems are
  improved immediately by adopting the proposed subdivision: 
  (1) computing the shortest path between two
  points inside a simple $n$-gon, (2) computing the shortest path tree
  from a point inside a simple $n$-gon, (3) computing a triangulation
  of a simple $n$-gon.  In addition, we improve the algorithm for
  problem (2) further by applying different approaches
  depending on the size of the workspace. 
\end{abstract}

\section{Introduction}
In the algorithm design for a given task, we seek to construct an
\emph{efficient} algorithm with respect to the time and space complexities.
However, one cannot achieve both goals at the same time in many cases:
one has to use more memory space for storing information necessary 
to achieve a faster algorithm and spend more time if less amount of memory is allowed.
Therefore, one has to make a compromise between the time and
space complexities, considering the goal of the task and the system
resources where the algorithm under design is performed.  With this reason, a
number of time-space trade-offs were considered even as early as in
1980s.  For example, Frederickson~\cite{Frederickson-Upperbounds-1987}
presented optimal time-space trade-offs for sorting and selection
problems in 1987. After this work, a significant amount of research
has been done for time-space trade-offs in the design of algorithms.

The model we consider in this paper is formally described as follows.
An input is given in a read-only memory.  For a positive integer
parameter $s$ which is determined by users, a memory space of
$O(s)$ words are available as workspace 
(read-write memory under a \emph{random access model})
in addition to the memory for input. 
We assume that a word is large enough to store a number
or a pointer. During the process, the output is to
be written to a write-only stream without repetition.
We assume that input is given in a \emph{read-only memory} under a
random-access model.  The assumption on the read-only memory
has been considered in applications where the input is required to be
retained in its original state or more than one program access the
input simultaneously.  An algorithm designed in this setting is
called an \emph{$s$-workspace algorithm}.
It is generally assumed that $s$ is sublinear in the size of input. 


Many classical algorithms 
require workspace of at least the size of input
in addition to the memory for input. However, 
this is not always possible
because the amount of data collected and used by various applications has
significantly increased over the last years and the memory resource
available in the system gets relatively smaller compared to the
amount of data they use.
The $s$-workspace algorithms deal with the case that the size of workspace is limited. Thus we assume that $s$ is at most the size of input throughout this paper.

\subsection{Previous Work}
In this paper, we consider time-space trade-offs for constructing a
few geometric structures inside a simple polygon: the shortest path
between two points, the shortest path tree from a point, and a
triangulation of a simple polygon.  With linear-size workspace,
optimal algorithms for these problems are known.  The shortest path
between two points and the shortest path tree from a point inside a
simple $n$-gon can be computed in $O(n)$
time using $O(n)$ words of workspace~\cite{guibas_linear_1987}.
A triangulation of a simple $n$-gon
can also be computed in $O(n)$
time using $O(n)$ words of workspace~\cite{Chazelle-triangulating-1991}.

For a positive integer parameter $s$, the following $s$-workspace
algorithms are known.
\begin{itemize}
\item \textbf{The shortest path between two points inside a simple
    polygon:} The first non-trivial $s$-workspace algorithm for
  computing the shortest path between any two points in a simple
  $n$-gon was given by Asano et
  al.~\cite{asano_memory-constrained_2013}.  Their algorithm consists
  of two phases. In the first phase, they subdivide the input simple
  polygon into $O(s)$ subpolygons of complexity $O(n/s)$ in $O(n^2)$
  time.  In the second phase, they compute the shortest path between
  the two points in $O(n^2/s)$ time using the subdivision.  In the
  paper, they asked whether the first phase can
  be improved to take $O(n^2/s)$ time.  This
  problem is still open while there are several partial results.
	
  Har-Peled~\cite{har-peled_shortest_2015} presented an $s$-workspace
  algorithm which takes $O(n^2/s+n\log s \log^4(n/s))$ expected time.
  Their algorithm takes $O(n^2/s)$ expected time for the case of
  $s=O(n/\log^2 n)$.  For the case that the input polygon is monotone,
  Barba et al.~\cite{barba_spacetime_2015} presented an $s$-workspace
  algorithm which takes $O(n^2/s+(n^2\log n)/2^s)$ time.  Their
  algorithm takes $O(n^2/s)$ time for $\log\log n\leq s <n$.
	
\item \textbf{The shortest path tree from a point inside a simple
    polygon:} The shortest path tree from a point $p$ inside a simple polygon
    is defined as the union of the shortest paths from $p$ to all vertices of the simple polygon.
Aronov et al.~\cite{aronov_time-space} presented an
  $s$-workspace algorithm for computing the shortest path tree from a
  given point.  Their algorithm reports the edges of the shortest path
  tree without repetition in an arbitrary order in
  $O((n^2\log n)/s + n\log s\log^5(n/s))$ expected time.
	
\item \textbf{A triangulation of a simple polygon:} Aronov et
  al.~\cite{aronov_time-space} presented an $s$-workspace algorithm
  for computing a triangulation of a simple $n$-gon.  Their algorithm
  returns the edges of a triangulation without repetition in
  $O(n^2/s + n\log s\log^5{(n/s)})$ expected time.  Moreover, their
  algorithm can be modified to report the resulting triangles of a
  triangulation together with their adjacency information in the same
  time if $s\geq \log n$.
	
  For a monotone $n$-gon, Barba et al.~\cite{barba_spacetime_2015}
  presented an $O(s\log_s n)$-workspace algorithm for triangulating the
  polygon in $O(n\log_s n)$ time for a parameter $s\in\{1,\ldots,n\}$.
  Later, Asano and Kirkpatrick~\cite{Asano-time-space-2013} showed how
  to reduce the workspace to $O(s)$ words without increasing the
  running time.
\end{itemize}

\subsection{Our Results}
In this paper, we present an $s$-workspace algorithm to subdivide a simple polygon
with $n$ vertices into $O(\min\{n/s,s\})$ subpolygons of complexity
$O(\max\{n/s,s\})$ in $O(n^2/s)$ deterministic time.  We obtain this
subdivision in three steps. First, we choose every $\max\{n/s,s\}$th
vertex of the simple polygon which we call \emph{partition vertices}.
In the second step, for every pair of consecutive partition vertices
along the polygon boundary,
we choose $O(1)$ vertices which we call \emph{extreme vertices}.  Then
we draw the vertical extensions from each partition vertex and each
extreme vertex, one going upwards and one going downwards, until the
extensions escape from the polygon for the first time. 
These extensions subdivide the polygon into subpolygons. In the subdivision,
however, some subpolygons may still have complexity strictly larger than 
$O(\max\{n/s,s\})$.
In the third step,
we subdivide each such subpolygon further into subpolygons of complexity
$O(\max\{n/s,s\})$.  Then we show that the resulting subdivision has
the desired complexity.

By using this subdivision method, we improve
the running times for the following three problems without increasing
the size of the workspace.
\begin{itemize}
\item \textbf{The shortest path between two points inside a simple
    polygon:} We can compute the shortest path between any two points
  inside a simple $n$-gon in $O(n^2/s)$ deterministic time using
  $O(s)$ words of workspace.  The previously best known $s$-workspace
  algorithm~\cite{har-peled_shortest_2015} takes
  $O(n^2/s + n\log s \log^4(n/s))$ expected time.
\item \textbf{The shortest path tree from a point inside a simple
    polygon:} 
  The previously best known $s$-workspace
  algorithm~\cite{aronov_time-space} takes
  $O((n^2\log n)/s + n\log s\log^5{(n/s)})$ expected time. It 
  uses the algorithm in~\cite{har-peled_shortest_2015} as a subprocedure for 
  computing the shortest path between two points. 
    If the subprocedure is replaced by our shortest path algorithm, 
  the algorithm is improved to take
  $O((n^2\log n)/s)$ expected time.
\item \textbf{A triangulation of a simple polygon:} The previously
  best known $s$-workspace algorithm~\cite{aronov_time-space} takes
  $O(n^2/s + n\log s \log^4(n/s))$ expected time, which 
  uses the shortest path algorithm in~\cite{har-peled_shortest_2015} as a subprocedure. 
  If the subprocedure is replaced by our shortest path algorithm, 
  the triangulation algorithm is improved to take only $O(n^2/s)$ deterministic time.
\end{itemize}

We also improve the algorithm for computing the shortest path tree from a given point
even further to take $O(n^2/s + (n^2\log n)/s^c)$ expected time for an arbitrary positive constant $c$.
The improved result is based on the constant-workspace algorithm by Aronov et al.~\cite{aronov_time-space} 
for computing the shortest path
tree rooted at a given point. Depending on the size of workspace, 
we use two different approaches. For the case of $s=O(\sqrt{n})$, we decompose the polygon
into subpolygons, each associated with a vertex, and for each subpolygon we compute the shortest path tree 
rooted at its associated vertex inside the subpolygon recursively. Due to the workspace constraint, we stop
the recursion at a constant depth once one of the stopping criteria is satisfied. Then we show how to
report the edges of the shortest path tree without repetition efficiently using $O(s)$ words of workspace.
For the case of $s=\Omega(\sqrt{n})$, we can store all edges of each subpolygon in the workspace. We decompose
the polygon into subpolygons associated with vertices and solve each subproblem directly using the algorithm
by Guibas et al.~\cite{guibas_linear_1987}.

\section{Preliminaries}
Let $P$ be a simple polygon with $n$ vertices.  Let
$v_0,\ldots,v_{n-1}$ be the vertices of $P$ in clockwise order along
$\bd P$.  The vertices of $P$ are stored in a read-only memory in this
order. For a subpolygon $S$ of $P$, we use $\bd S$ to denote the
boundary of $S$ and $|S|$ to denote the number of vertices of $S$.  For any
two points $p$ and $q$ in $P$, we use $\pi(p,q)$ to denote the
shortest path between $p$ and $q$ contained in $P$.  To ease the
description, we assume that no two distinct vertices of $P$ have the
same $x$-coordinate.  We can avoid this assumption by using a shear
transformation~\cite[Chapter 6]{CGbook}.

Let $v$ be a vertex of $P$. We consider two vertical extensions from
$v$, one going upwards and one going downwards, until they escape from
$P$ for the first time.  A vertical extension from $v$ contains no
vertex of $P$ other than $v$ due to the assumption we made above.
We call the point of $\bd P$ where an extension from $v$ escapes from
$P$ for the first time a \emph{foot point} of $v$.  Note that a foot
point of a vertex might be the vertex itself.  
The following two lemmas show how to compute and report the foot points
of vertices using $O(s)$ words of workspace 
\begin{lemma}
  \label{lem:compute-extreme-small}
  For a polygonal chain $\gamma\subseteq\bd P$ of size $O(s)$, we can compute the
  foot points of all vertices of $\gamma$ in $O(n)$ deterministic time
  using $O(s)$ words of workspace.
\end{lemma}
\begin{proof}
  We show how to compute the foot point of every vertex $v$ of
  $\gamma$ lying above $v$ only. The other foot points can be computed
  analogously.  The foot point of a vertex $v$ of $\gamma$ might be
  $v$ itself. We can determine whether the foot point of $v$ is $v$
  itself or not in $O(1)$ time by considering the two edges incident
  to $v$.  
	
  We split the boundary of $P$ into $O(n/s)$ polygonal chains each of
  which contains $O(s)$ vertices. Let $\beta_0,\ldots,\beta_t$ be the
  resulting polygonal chains with $t=O(n/s)$.  For a vertex
  $v\in\gamma$ whose foot point is not $v$ itself, let $\beta_i(v)$
  denote the first point of $\beta_i$ (excluding $v$) 
  hit by the upward vertical ray 
  from $v$ for each $i=0,\ldots,t$. If there is no such point, we let $\beta_i(v)$ denote a
  point at infinity.  We observe that the foot point of a vertex
  $v\in\gamma$ is the one closest to $v$ among $\beta_i(v)$'s for
  $i=0,\ldots,t$ unless its foot point is $v$ itself.
	
  For any fixed index $i\in\{0,\ldots, t\}$, we can compute $\beta_i(v)$ for all
  vertices $v\in\gamma$ whose foot points are not themselves in $O(s)$
  time using $O(s)$ words of workspace using the algorithm
  in~\cite{chazelle_triangulation_1984}.  This algorithm computes the
  vertical decomposition of a simple polygon in linear time using
  linear space, but it can be modified to compute the vertical
  decomposition of any two non-crossing polygonal curves without
  increasing the time and space complexities.
  Since both $\beta_i$ and $\gamma$ have size of $O(s)$,
  we can apply the vertical decomposition algorithm in~\cite{chazelle_triangulation_1984}
  in $O(s)$ time using $O(s)$ words of workspace.
	
  We apply this algorithm to $\beta_0$. For each vertex $v$ of
  $\gamma$ whose foot point is not $v$ itself, we store $\beta_0(v)$
  in the workspace.  Now we assume that we have the one closest to $v$
  among $\beta_i(v)$'s, for $i=0,\ldots,j-1$, stored in the workspace. To compute the one
  closest to $v$ among $\beta_i(v)$'s for $i=1,\ldots,j$, we compute
  $\beta_j(v)$.  This can be done in $O(s)$ time for all vertices on
  $\gamma$ whose foot points are not themselves using the algorithm
  in~\cite{chazelle_triangulation_1984}. Then we compare $\beta_j(v)$
  and the one stored in the workspace, 
  choose the one closer to $v$ between them and store it in the workspace. 
  
  Once we do this for all polygonal chains $\beta_i$, we obtain the
  foot points of all vertices of $\gamma$ by the observation. Since we
  spend $O(s)$ time for each polygonal chain $\beta_i$, the total running time is
  $O(n)$.
\end{proof}

\begin{lemma}
	\label{lem:compute-footpoints}
	We can report the foot points of all vertices of $P$ 
	in $O(n^2/s)$ deterministic time using $O(s)$ words of workspace. 
\end{lemma}
\begin{proof}
	We apply the procedure in Lemma~\ref{lem:compute-extreme-small} 
	to the first $s$ vertices of $P$, the next $s$ vertices
	of $P$, and so on. In this way we apply this procedure $O(n/s)$ times.
	Thus we can find all foot points in $O(n^2/s)$ time.
\end{proof}

The extensions from some vertices of $P$ induce a subdivision of $P$
into subpolygons. 
Notice that the number of subpolygons in the subdivision
is linear to the number of extensions. In the following sections, we compute
$O(\min \{n/s, s\})$ extensions from vertices of $P$ and use them to subdivide
$P$ into $O(\min \{n/s, s\})$ subpolygons.
We store the endpoints of the extensions of the subdivision together with
the extensions themselves in clockwise order along $\bd P$ in the
workspace. Then
we can traverse the boundary of the subpolygon starting from a given edge of
the subpolygon in time linear to the complexity
of the subpolygon. 

\section{Balanced Subdivision of a Simple Polygon}
\label{sec:subdivision}
We say that a subdivision of $P$ with $n$ vertices
\emph{balanced} if it subdivides $P$ into $O(\min\{n/s,s\})$ subpolygons
of complexity $O(\max\{n/s,s\})$.
In this section, we present an $s$-workspace algorithm that computes
a balanced subdivision using $O(\min\{n/s,s\})$ extensions in $O(n^2/s)$ time.
In the following sections, we present
a subdivision procedure in three steps. Then we show that the subdivision
is balanced.

\begin{figure}
  \begin{center}
    \includegraphics[width=0.7\textwidth]{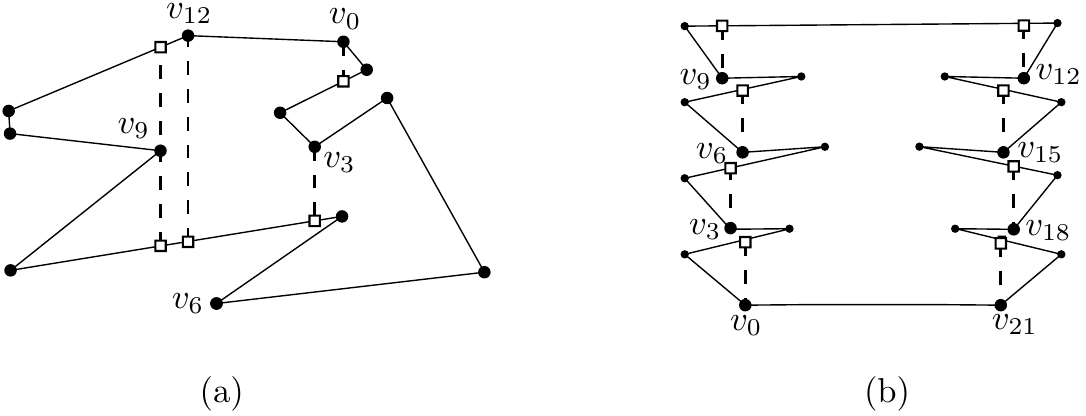}
    \caption{\label{fig:subdivide}\small (a) Subdivision of $P$ with
      $\triangle=3$ induced by partition vertices.  (b) The subpolygon
      in the middle is incident to 
      $n/\triangle$ vertical extensions, and therefore it has complexity 
      strictly larger than $O(\triangle)$
        for a constant $\triangle$.}
  \end{center}
\end{figure}
\subsection{Subdivision in Three Steps}
We first present an $s$-workspace algorithm to subdivide $P$ into
$O(n/\triangle)$ subpolygons of complexity $O(\triangle)$ using
$O(n/\triangle)$ extensions in $O(n^2/s)$ time, where $\triangle$ is a
positive integer satisfying
$\max\{n/s, (s\log n)/n\}\leq \triangle \leq n$
which is determined by $s$.  Since
  $n/s\leq \triangle$, we have $n/\triangle \leq s$.
Thus, we can keep all such extensions in the
workspace of size $O(s)$.  We will set the value of $\triangle$ in
Theorem~\ref{thm:subdivide} so that we can obtain a subdivision of our
desired complexity.
\label{sec:balanaced-subdivision-3steps}
\paragraph{The first step: Subdivision by partition vertices.}
We first consider every $\triangle$th vertex of $P$ from $v_0$ in
clockwise order, that is,
$v_0,v_{\triangle},v_{2\triangle},\ldots,v_{\lfloor
  n/\triangle\rfloor\triangle}$.  We call them
\emph{partition vertices}.  The number of partition vertices is
$O(n/\triangle)$.  We compute the foot points of each
partition vertex, which can be done for all partition vertices in
$O(n^2/s)$ time in total using $O(s)$ words of workspace by
Lemma~\ref{lem:compute-footpoints}.  We sort the foot points along
$\bd P$ in $O((n/\triangle) \log (n/\triangle))$ time, which is
$O(n^2/s)$ by the fact that $\triangle\geq (s\log n)/n$. We
store them together with their vertical extensions using
$O(n/\triangle)=O(s)$ words of workspace.

The vertical extensions of the partition vertices subdivide $P$ into
$O(n/\triangle)$ subpolygons.  See Figure~\ref{fig:subdivide}(a).
However, there might be a subpolygon with complexity
  strictly larger than $O(\triangle)$.
See Figure~\ref{fig:subdivide}(b).
Recall that our goal is to subdivide $P$ into $O(n/\triangle)$
subpolygons each of complexity $O(\triangle)$.  To achieve this
complexity, we subdivide each subpolygon further.

\paragraph{The second step: Subdivision by extreme vertices.}
The \emph{\lc-extreme vertex} and \emph{\lcc-extreme vertex} of a
polygonal chain $\gamma$ of $\bd P$ are defined as follows.
Let $V_\gamma$ be the set of all vertices of $\gamma$ both of whose
foot points are on $\bd P\setminus \gamma$ and whose extensions lie
locally to the \emph{left} of $\gamma$.  The \lc-extreme vertex (or
the \lcc-extreme vertex) of $\gamma$ is the vertex in $V_\gamma$
defining the first extension we encounter while we traverse $\bd P$ in
clockwise (or counterclockwise) order from $v_0$.  See
Figure~\ref{fig:second-third}(a) for an illustration.  Similarly, we
define the \emph{\rc-extreme vertex} and \emph{\rcc-extreme vertex} of
$\gamma$.  In this case, we consider the vertices of $\gamma$ whose
extensions lie locally to the \emph{right} of $\gamma$.  We simply
call the \lc-,\lcc-,\rc- and \rcc-extreme vertices
\emph{extreme vertices} of $\gamma$.  Note that $\gamma$ may not have
any extreme vertex.

In the second step, we consider every polygonal
chain of $\bd P$ connecting two consecutive partition vertices along
$\bd P$ and compute the extreme vertices of the chain. 
Then we have $O(n/\triangle)$
extreme vertices. We compute the foot points of all extreme vertices
and store them together with their vertical extensions using
$O(n/\triangle)=O(s)$ words of workspace in $O(n^2/s)$ time using
Lemma~\ref{lem:compute-footpoints} and
Lemma~\ref{lem:compute-extreme}.

\begin{figure}
  \begin{center}
    \includegraphics[width=0.75\textwidth]{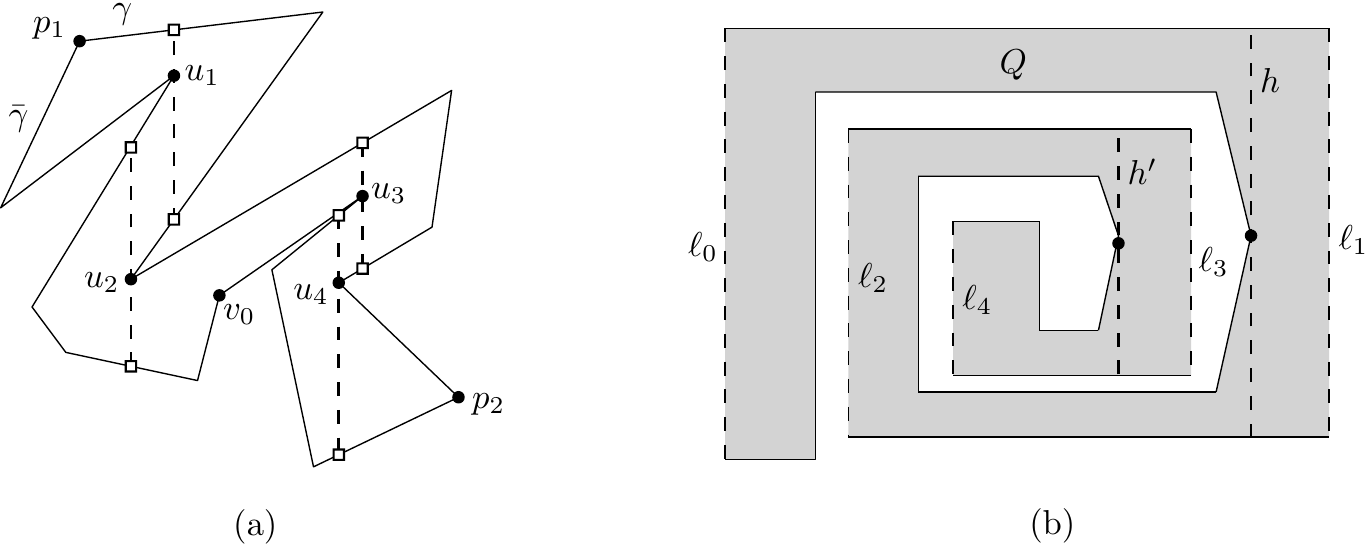}
    \caption{\label{fig:second-third} \small (a) Two chains $\gamma$
      and $\bar{\gamma}$ connecting two vertices $p_1$ and $p_2$.  The
      set $V_{\gamma}=\{u_2,u_4\}$.  The \lc-extreme vertex of
      $\gamma$ is $u_2$, and the \lcc-extreme vertex of $\gamma$ is
      $u_4$.  The \rc-extreme vertex of $\bar{\gamma}$ is $u_1$, and
      the \rcc-extreme vertex of $\bar{\gamma}$ is $u_3$.  (b) In the
      third step, we compute the vertical extension $h$ for
        $(\ell_0,\ell_1,\ell_2)$ and the vertical extension $h'$ for
        $(\ell_2,\ell_3,\ell_4)$ that subdivide $\subp$ into five
        subpolygons.}   \end{center}
\end{figure}

\begin{lemma}
  \label{lem:compute-extreme}
  We can find the extreme vertices of every polygonal chain of $\bd P$
  connecting two consecutive partition vertices along $\bd P$ in $O(n^2/s)$
  total time using $O(s)$ words of workspace.
\end{lemma}
\begin{proof}
  Let $\beta_i$ be the polygonal chain of $\bd P$ connecting two
    consecutive partition vertices $v_{i\triangle}$ and
    $v_{(i+1)\triangle}$ ($v_{i\triangle}$ and $v_0$ if
    $i=\lfloor n/\triangle\rfloor$) along $\bd P$ for
  $i=0,\ldots, \lfloor n/\triangle\rfloor$.  We show how to compute
  the \lc-extreme vertices of $\beta_i$ for all $i$.  The other
  types of extreme vertices can be computed analogously.
	
  We apply the algorithm in Lemma~\ref{lem:compute-footpoints} that
  reports the foot points of every vertex of $P$.  During the
  execution of the algorithm, for every $i$, we store one vertex for
  $\beta_i$ together with its foot points as a candidate of the
  $\lc$-extreme vertex of $\beta_i$.  These vertices are updated
  during the execution of the algorithm. At the end of the execution,
  we guarantee that the vertex stored for $\beta_i$ is the
  \lc-extreme vertex of $\beta_i$ for every $i$ from 0 to $\lfloor n/\triangle\rfloor$.
	
  Assume that the algorithm in Lemma~\ref{lem:compute-footpoints}
  reports the foot points of a vertex $v\in\beta_i$.  If the
  extensions of $v$ lie locally to the left of $\beta_i$, we update
  the vertex for $\beta_i$ as follows.  We compare $v$ and the vertex
  $v'$ stored for $\beta_i$. Specifically, we check if we encounter the extension of $v$
  before the extension of $v'$ during the traversal of $\bd P$ from
  $v_0$ in clockwise order.
  We can check this in constant time
  	using the foot points of $v'$ which are stored for $\beta_i$ together with $v'$. If so, we store $v$ for $\beta_i$ together with
  its foot points instead of $v'$.  Otherwise, we just keep $v'$ for
  $\beta_i$.
	
  In this way, for every chain $\beta_i$, we consider the foot points
  of all vertices on $\beta_i$ whose extensions lie to the left of
  $\beta_i$, and keep the extension which comes first from $v_0$ in
  clockwise order.  Thus, at the end of the algorithm, we have the
  \lc-extreme vertex of every polygonal chain $\beta_i$ by
  definition.  This takes $O(n^2/s)$ time in total, which is the time
  for computing the foot points of all vertices of $P$ by
  Lemma~\ref{lem:compute-footpoints}.
\end{proof}

\paragraph{The third step: Subdivision by a vertex on a chain
  connecting three extensions.}
After applying the first and second steps, we obtain the subdivision
induced by the extensions from the partition and extreme vertices.
Let $\subp$ be a subpolygon in this subdivision.  We will see later in
Lemma~\ref{lem:no-consecutive} that $\subp$ has the following
property: every chain connecting two consecutive extensions along
$\bd \subp$ has no extreme vertex, except for two such chains.

However,
it is still possible that $\subp$ contains 
more than a constant number of extensions on its boundary.
For instance, Figure~\ref{fig:second-third}(b) shows a spiral-like subpolygon
in the subdivision constructed after the first and second steps that has five
extensions on its boundary. The input polygon can easily be modified to
have more than a constant number of extensions on the boundary of such 
a spiral-like subpolygon. 

In the third step, we subdivide each subpolygon further so that every subpolygon has
$O(1)$ extensions on its boundary.
The boundary of $\subp$ consists of vertical extensions and polygonal
chains from $\bd P$ whose endpoints are partition vertices,
extreme vertices, or their foot points.  We treat the upward
  and downward extensions defined by one partition  or extreme vertex
(more precisely, the union of them) as one vertical extension.

For every triple $(\ell,\ell',\ell'')$ of consecutive vertical
extensions appearing along $\bd
  \subp$ 
in clockwise order, we consider the part (polygonal chain) of
$\bd \subp$ from $\ell$ to $\ell''$ in clockwise order (excluding
$\ell$ and $\ell''$).  Let $\Gamma$ be the set of all such polygonal
chains.  For every $\gamma\in\Gamma$, we find a vertex, denoted by $v(\gamma)$, of
$\bd \subp\setminus \gamma$ such that one of its foot points lies in
$\gamma$ between $\ell$ and $\ell'$, and the other foot point lies in
$\gamma$ between $\ell'$ and $\ell''$ if it exists.  If there are more
than one such vertex, we choose an arbitrary one.

The extensions of $v(\gamma)$ subdivide $\subp$ into three subpolygons
each of which contains one of $\ell, \ell'$ and $\ell''$ on its boundary.  In other
words, the extensions from $v(\gamma)$ \emph{separate} $\ell, \ell'$
and $\ell''$.  In Figure~\ref{fig:second-third}(b), the
  vertical extension $h$ for $(\ell_0,\ell_1,\ell_2)$ and the vertical
  extension $h'$ for $(\ell_2,\ell_3,\ell_4)$ together subdivide
  $\subp$ into five subpolygons.  We can compute $v(\gamma)$ and
their extensions for every $\gamma\in\Gamma$ in
$O(|\subp|^2/s + m(\subp))$ time in total, where $m(\subp)$ denotes the number of the extensions on
the boundary of $Q$.

\begin{lemma}
  \label{lem:compute-extreme-complement}
  We can find $v(\gamma)$ for every $\gamma\in\Gamma$ in
  $O(|\subp|^2/s+m(\subp))$ total time using $O(s)$ words of
  workspace.
\end{lemma}
\begin{proof}
  The algorithm is similar to the one in
  Lemma~\ref{lem:compute-extreme}.  We apply the algorithm in
  Lemma~\ref{lem:compute-footpoints} to compute the foot points of
  every vertex of $\subp$ with respect to $\subp$.  Assume that the
  algorithm in Lemma~\ref{lem:compute-footpoints} reports the foot
  points of a vertex $v$ of $\subp$. We find the polygonal chains
  $\gamma\in\Gamma$ containing both foot points of $v$ if they exist.
  There are at most two such polygonal chains by the
    construction of $\Gamma$.  We can find them in constant time
  after an $O(m(\subp))$-time preprocessing for $\subp$ by
  Lemma~\ref{lem:point-chain}.  Let $\ell,\ell'$ and $\ell''$ be the
  three extensions defining $\gamma$.  Then we check whether one foot
  point of $v$ lies on the part of $\gamma$ between $\ell$ and
  $\ell'$, and the other foot point of $v$ lies on the part of
  $\gamma$ between $\ell'$ and $\ell''$.  If so, we denote this vertex
  by $v(\gamma)$ and keep it for $\gamma$.  Otherwise, we do
  nothing. In this way, we can find $v(\gamma)$ if it exists since we
  consider every vertex whose foot points lie on $\gamma$.  This takes
  $O(|\subp|^2/s+m(\subp))$ time in total, which is the time for
  computing the foot points of all vertices of $\subp$ plus the
  preprocessing time for $\subp$.
\end{proof}

\begin{lemma}
  \label{lem:point-chain}
  For any point $p$ on $\bd \subp$, we can find the polygonal chains in
  $\Gamma$ containing $p$ in constant time, if they exist, after an
  $O(m(\subp))$-time preprocessing for $\subp$,
  where $m(\subp)$ denotes the number of the extensions on the boundary of $Q$.
\end{lemma}
\begin{proof}
  Imagine that we 
  subdivide $\bd P$ with respect to the partition vertices of $P$ into
  $O(n/\triangle)$ chains. 
  Each chain $\beta$ in the subdivision of $\bd P$
  intersects at most two chains $f_1(\beta),f_2(\beta)\in\Gamma$
  by the construction of $\Gamma$.  As a preprocessing, for each
  chain $\beta$ in the subdivision of $\bd P$ by
    the partition vertices, we store $f_1(\beta)$ and $f_2(\beta)$.
  There are $O(n/\triangle)$ chains of $\bd P$,
  but only $O(m(\subp))$ of them have
    their $f_1(\cdot)$ and $f_2(\cdot)$.
  Thus, we can find and store for
  all such chains their $f_1(\cdot)$ and $f_2(\cdot)$ in $O(m(\subp))$ time as follows.
  For each $\gamma\in\Gamma$, we find
  two chains $\beta_1$ and $\beta_2$ of $\bd P$
  containing $\gamma$ in constant time, and set $f_i(\beta_1)=\gamma$ and
  $f_i(\beta_2)=\gamma$ for $i=1,2$, accordingly.
	
  For any point $q$ on $\bd P$, we can find the subchain $\beta$ in
  the subdivision of $\bd P$ containing $q$ in constant time because
  the partition vertices are distributed uniformly 
  at intervals of $\triangle$ vertices along $\bd P$. Then we check
  whether $f_1(\beta)$ and $f_2(\beta)$ contain $p$ in constant time.
\end{proof}

The sum of $|\subp|$ over all subpolygons $\subp$ is $O(n)$ and the number
of the subpolygons from the second step is $O(n/\triangle)$ since we
construct $O(n/\triangle)$ extensions in the first and second
steps. Therefore, we can apply the third step of the subdivision 
for all subpolygons in the subdivision from the second step in
$O(n^2/s+n)=O(n^2/s)$ time using $O(s)$ words of workspace.

\subsection{Balancedness of the Subdivision}
\label{sec:balanced-subdivision-analysis}
We obtained $O(n/\triangle)$ vertical extensions in $O(n^2/s)$ time
using $O(s)$ words of workspace.  In this section, we show that these
vertical extensions subdivide $P$ into $O(n/\triangle)$ subpolygons of
complexity $O(\triangle)$.  We call this subdivision the
\emph{balanced subdivision} of $P$.
For any two points $a, b$ on $\bd P$, we use $P[a,b]$ to denote the
polygonal chain from $a$ to $b$ (including $a$ and $b$) in clockwise
order along $\bd P$. 

We use a few technical lemmas (Lemma~\ref{lem:contain-par} to
Lemma~\ref{lem:num-third}) to show that each subpolygon in the final
subdivision is incident to $O(1)$ extensions and has complexity of
$O(\triangle)$.  Then we obtain Theorem~\ref{thm:path} by setting a
parameter $\triangle$.

\begin{figure}
  \begin{center}
    \includegraphics[width=0.8\textwidth]{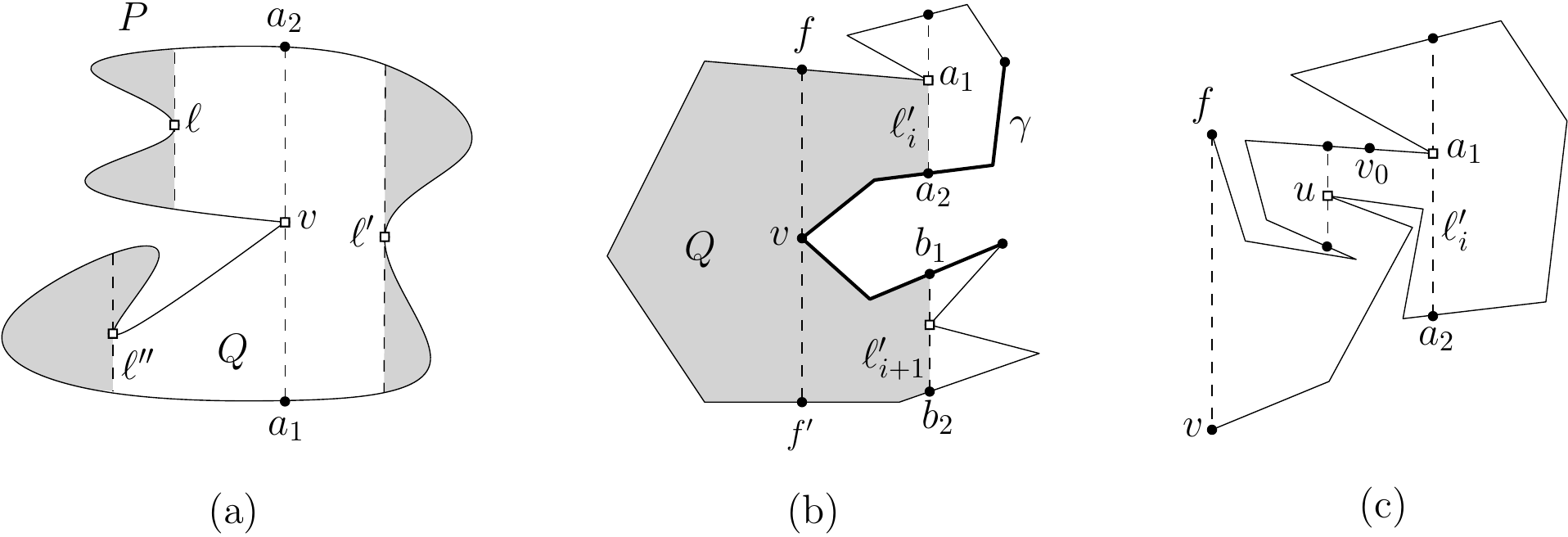}
    \caption{\label{fig:analysis-consecutive} \small (a) Each
    	gray region has a partition vertex on its boundary.
    	(b) If
      $P[a_2,b_1]$ has an \lc- or \lcc-extreme vertex, $v_0$ lies
      on $P[a_1,b_2]$, which implies that $i=0$ or $i=k'$.  (c) If
      $v_0$ lies on $P[f,a_1]$, an extension which separates $\ell_i'$
      and $\ell_{i+1}'$ is constructed in the second step.  }
  \end{center}
\end{figure}

\begin{lemma}\label{lem:contain-par}
  Let $a_1a_2$ be any extension constructed from a vertex $v$
    during any of the three steps such that $P[a_1,a_2]$ contains $v$.
    Then both $P[a_1,v]$ and $P[v,a_2]$ contain partition vertices.
\end{lemma}
\begin{proof}
  If $a_1a_2$ is constructed in the first step, $v$ is a
  partition vertex and lies on $P[a_1,v]$ and $P[v,a_2]$, and we are done.
  If $a_1a_2$ is constructed in the second step, $v$ is an extreme vertex of a
  polygonal chain which connects two consecutive
  partition vertices. One of the two partition vertices lies on
  $P[a_1,v]$ and the other lies on $P[v,a_2]$, thus the claim holds.
        
  Now, consider the case that $a_1a_2$ is constructed in the third
  step.  In this case, $a_1a_2$ separates three consecutive extensions
  $\ell,\ell'$ and $\ell''$ which are constructed in the first or second step
  of the subdivision. See Figure~\ref{fig:analysis-consecutive}(a). 
  
  Let $Q$ be the subpolygon of $P$ bounded
  by the three extensions.  
  Then every connected component of $P\setminus Q$ contains a partition vertex on
  its boundary contained in $\bd P$ because each component is incident to an extension
  constructed in the first or second step.
  In Figure~\ref{fig:analysis-consecutive}(a),
  the component of $P\setminus Q$ incident to $\ell''$ has a partition vertex on its boundary contained 
  in $P[a_1,v]$.
  Similarly, the component of $P\setminus Q$ incident to $\ell$ has a partition vertex on
  its boundary contained in $P[v,a_2]$.
  Thus, both $P[a_1,v]$ and $P[v,a_2]$ contain partition vertices.
\end{proof}

Let $S$ be a subpolygon in the final subdivision and $\subp$ be the
subpolygon in the subdivision from the second step containing $S$. We
again treat the two (upward and downward) vertical extensions defined
by one vertex as one vertical extension.  We label the extensions
lying on $\bd S$ as follows.  Let $\ell_0$ be the first extension on
$\bd S$ we encounter while we traverse $\bd P$ from $v_0$ in clockwise
order.  We let $\ell_1,\ell_2,\ldots, \ell_k$ be the extensions
appearing on $\bd S$ in clockwise order along $\bd S$ from $\ell_0$.
Similarly, we label the extensions lying on $\bd \subp$ from $\ell_0'$
to $\ell_{k'}'$ in clockwise order along $\bd \subp$ such that
$\ell_0'$ is the first one we encounter while we traverse $\bd P$ from
$v_0$ in clockwise order.  Then we have the following lemmas.

\begin{lemma}
  \label{lem:no-consecutive}
  For any $1\leq i<k'$, let $a_1a_2=\ell_i'$ and $b_1b_2=\ell_{i+1}'$
  such that $a_1, a_2, b_1$ and $b_2$ appear on $\bd P$ (and
  on $\bd \subp$) in clockwise order.  Then $P[a_2,b_1]$ has no
  extreme vertex.
\end{lemma}
\begin{proof}
  Assume to the contrary that for some $1\leq i<k'$, $P[a_2,b_1]$ has
  an extreme vertex. 
  For an illustration, see Figure~\ref{fig:analysis-consecutive}(b). 
  By definition, no partition vertex lies on
  $P[a_2,b_1]\setminus\{a_2,b_1\}$.  Consider the maximal polygonal
  chain $\gamma\subset \bd P$ containing no partition vertex in its
  interior and containing $P[a_2,b_1]$. Note that
  $\gamma\subseteq P[a_1,b_2]$ since both $P[a_1,a_2]$ and
  $P[b_1,b_2]$ contain partition vertices by
  Lemma~\ref{lem:contain-par}.

  Let $v$ be an extreme vertex of $P[a_2,b_1]$. (Recall that $v$
  exists by the assumption made in the beginning of the proof.)
  Without loss of generality, we assume that $P[a_2,b_1]$ lies locally
  to the right of the extension of $v$.  The foot points of $v$ lie on
  $\bd P\setminus \gamma$ while $v$ lies on $\gamma$.  Therefore,
  $\gamma$ has an extreme vertex. (But $v$ is not necessarily an
  extreme vertex of $\gamma$ by definition.)  The
  extension of $v$ subdivides $P$ into three subpolygons.  Let $f$ be
  the foot point of $v$ incident to the subpolygon containing
  $\ell_i'$ on its boundary and $f'$ be the other foot point of $v$,
  as shown in Figure~\ref{fig:analysis-consecutive}(b).

  Since $1\leq i<k'$, $v_0$ lies on $P[f',f]$, $P[f,a_1]$ or
  $P[b_2,f']$. (Recall
  that the vertices of $P$ are labeled from $v_0$ to $v_{n-1}$ in
  clockwise order.)  We show that for any case, there is an
  extreme vertex on $\gamma$ whose extension separates $\ell_i'$ and
  $\ell_{i+1}'$. Note that these extensions are constructed in the
  second step, which contradicts the assumption that $\subp$ contains
  both $\ell_i'$ and $\ell_{i+1}'$ on its boundary.
  \begin{itemize}
  \item \text{Case 1.} $v_0$ is in $P[f',f]$: Then $v$ is the \lc-
    and \lcc-extreme vertex of $\gamma$ by definition.  The
    extension of $v$ separates $\ell_i'$ and $\ell_{i+1}'$, which is a
    contradiction.
  \item \text{Case 2.} $v_0$ is in $P[f,a_1]$: By definition, the
    foot points of the \lcc-extreme vertex $u$ of $\gamma$ lie on
    $P[f,v_0]$. See Figure~\ref{fig:analysis-consecutive}(c).
    Moreover, $u$ lies on $P[a_2,v]$.  Thus, the extension of $u$
    separates $\ell_i'$ and $\ell_{i+1}'$, which is a contradiction.
  \item \text{Case 3.} $v_0$ is in $P[b_2,f']$: A contradiction can be shown
    in a way similar to Case 2. The only difference is that we
    consider the \lc-extreme vertex instead of the
    \lcc-extreme vertex.
  \end{itemize}
  Therefore, $P[a_2,b_1]$ has no extreme vertex.
\end{proof}

\begin{figure}
  \begin{center}
    \includegraphics[width=0.9\textwidth]{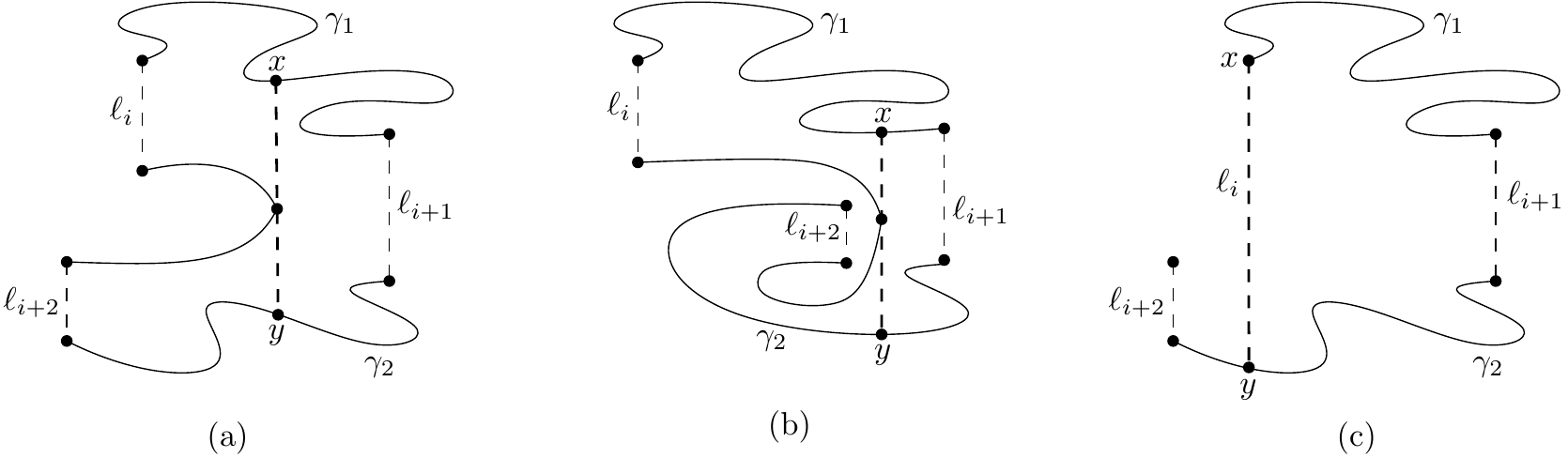}
    \caption{\label{fig:exist} \small (a-b) The segment $xy$
      intersects a point in $\bd \subp\setminus \gamma$ (and therefore in $\bd S\setminus \gamma$) 	
      in its interior.  (c) If
      $v(\gamma)$ does not exist, $xy$ coincides with $\ell_i$ or
      $\ell_{i+2}$. This is a contradiction. }
  \end{center}
\end{figure}
We need a few more technical lemmas, which are given in
the following, to conclude that
the subdivision proposed in the previous section is balanced.

\begin{lemma}
  \label{lem:one-of}
  For any $1\leq i<k-1$, one of $\ell_i,\ell_{i+1}$ and $\ell_{i+2}$
  is constructed in the third step.
\end{lemma}
\begin{proof}
  Assume to the contrary that all of $\ell_i,\ell_{i+1}$ and
  $\ell_{i+2}$ are constructed prior to the third step for some
  index $1\leq i<k-1$.  Then the three extensions are consecutive along
    $\bd \subp$ as well since there is no vertical extensions added to the
  part of $\bd Q$ from $\ell_i$ to $\ell_{i+2}$ in clockwise order in the third step.
Let $\gamma_1$ be the part of
$\gamma$ lying between $\ell_i$ and $\ell_{i+1}$ excluding the two extensions, and let $\gamma_2$
be the part of $\gamma$ lying between $\ell_{i+1}$ and $\ell_{i+2}$ excluding the two extensions.
  By Lemma~\ref{lem:no-consecutive}, $\gamma_1$ and $\gamma_2$ have no
  extreme vertex. 
  Thus, $\gamma_1\cup\ell_{i+1}\cup\gamma_2$ has no
  extreme vertex.
	
  We claim that $v(\gamma)$ exists. 
  Consider the point $x\in\gamma_1$ closest to an endpoint of $\ell_i$
  along $\gamma_1$ among the points in $\gamma_1$ one of whose foot
    points  is on $\gamma_2$.
   Let $y$ be the foot point of $x$ lying on $\gamma_2$. See Figure~\ref{fig:exist}(a-b). If $xy$ intersects
  some point in $\bd \subp\setminus \gamma$ (and therefore in $\bd S\setminus \gamma$) in its interior, such a point is
  $v(\gamma)$.  Otherwise, $xy$ coincides with $\ell_i$ or
  $\ell_{i+2}$. See
  Figure~\ref{fig:exist}(c).  This means that $\ell_i$ separates
  $\ell_{i+1}$ and $\ell_{i+2}$, or $\ell_{i+2}$ separates $\ell_i$
  and $\ell_{i+1}$.  This contradicts that $\ell_i,\ell_{i+1}$ and
  $\ell_{i+2}$ appear on $\bd S$ (and on $\bd \subp$).  Thus, the claim
  holds.
	
  In the third step, we construct the extensions of $v(\gamma)$, which
  separate $\ell_i,\ell_{i+1}$ and $\ell_{i+2}$.  This is a
  contradiction.
\end{proof}

\begin{lemma}
  \label{lem:num-third}
  $S$ has $O(1)$ extensions constructed in the third step on its boundary.
\end{lemma}
\begin{proof}
  Consider an extension $\ell$ incident to $S$ constructed in the
  third step.  Let $v$ be the vertex defining the extension $\ell$.
  Recall that the boundary of $\subp$ consists of the extensions
  $\ell_0',\ldots,\ell_{k'}'$ and the polygonal chains of $\bd P$ 
  connecting the pairs of the extensions in consecutive order. Let $\eta_i$ be the polygonal chain of
  $\bd P$ connecting $\ell_i'$ and $\ell_{i+1}'$, excluding the
  extensions, for $0\leq i<k'$, and $\eta_{k'}$ be the polygonal chain
  connecting $\ell_{k'}'$ and $\ell_0'$, excluding the extensions.

  We claim that $v$ is contained in $\eta_0$ or $\eta_{k'}$. Assume to
  the contrary that $v$ is contained in $\eta_i$ for $1\leq i<k'$. Then the
  foot points of $v$ lie outside of $\eta_i$ by the third step of the
  subdivision. Thus, $\eta_i$ has an extreme vertex, which
  contradicts Lemma~\ref{lem:no-consecutive}.
	
  We also claim that there exist at most two vertices in $\eta_0$
  that has both foot points in $\bd \subp\setminus\eta_0$
  and an extension incident to $S$. To see this, let
  $u_1, u_2\in\eta_0$ be such vertices if they exist.  Let $h_1$ and
  $h_2$ be the extensions from $u_1$ and $u_2$, respectively, incident
  to $S$. Since no foot point of $u_1$ and $u_2$ is in $\eta_0$,
   one of the two polygonal chains connecting $h_1$ and $h_2$ along
  $\bd S$ (but not containing them in its interior) is contained in
  $\eta_0$ and the other is disjoint with $\eta_0$. 
  Therefore, no other vertex in
  $\eta_0$ that has both foot points in $\bd S\setminus\eta_0$ and has an extension
  incident to $S$. This proves the claim. The same holds for
  $\eta_{k'}$.
	
  Therefore, there are at most four extensions on $\bd S$ constructed
  in the third step: two of them are extensions of vertices of
  $\eta_0$ and the other two are extensions of vertices of
  $\eta_{k'}$. Thus the lemma holds.
\end{proof}

Due to Lemma~\ref{lem:one-of} and Lemma~\ref{lem:num-third}, the
following corollary holds.
\begin{corollary}
  \label{lem:incident-vertical-line}
  Every subpolygon in the final subdivision has $O(1)$
  extensions on its boundary.
\end{corollary}

\begin{lemma}
  \label{lem:complexity}
  Every subpolygon in the final subdivision has complexity of
  $O(\triangle)$.
\end{lemma}
\begin{proof}
  Consider a subpolygon $S$ in the final subdivision.  By
  Corollary~\ref{lem:incident-vertical-line}, the boundary of $S$
  consists of $O(1)$ vertical extensions and $O(1)$ polygonal chains
  from the boundary of $P$ connecting two consecutive endpoints of
  vertical extensions along $\bd S$.  Each polygonal chain from the boundary
  of $P$ contains at most one partition vertex in its
  interior. Otherwise, a vertical extension intersecting the interior
  of $S$ is constructed in the first or second step, which contradicts
  that $S$ is a subpolygon in the final subdivision.  The number of
  vertices between two consecutive partition vertices along $\bd S$ is
  $O(\triangle)$.  Therefore, $S$ has $O(\triangle)$ vertices on its
  boundary.
\end{proof}

Therefore, we have the following lemma and theorem.
\begin{lemma}
  \label{lem:subdivide-parameter}
  Given a simple $n$-gon and a parameter $\triangle$ with
  $\max\{n/s, (s\log n)/n\}\leq \triangle \leq n$,
  we can compute a set of $O(n/\triangle)$ extensions
  which subdivides the polygon into $O(n/\triangle)$ subpolygons of
  complexity $O(\triangle)$ in $O(n^2/s)$ time using $O(s)$ words of
  workspace.
\end{lemma}

\begin{theorem}
	\label{thm:subdivide}
	Given a simple $n$-gon, 
	we can compute a set of $O(\min\{n/s,s\})$ extensions which subdivides the polygon into $O(\min\{n/s,s\})$ subpolygons of complexity 
	$O(\max\{n/s,s\})$ in $O(n^2/s)$ time using $O(s)$ words of workspace.	
\end{theorem}
\begin{proof}
	If $s\leq \sqrt{n}$, we set $\triangle$ to $n/s$. In this case, 
	we can subdivide the polygon into $O(s)$ subpolygons of complexity
	$O(n/s)$ by Lemma~\ref{lem:subdivide-parameter}.
	If $s> \sqrt{n}$, we set $\triangle$ to $s$. Note that $\max\{n/s, (s\log n)/n\}\leq \triangle \leq n$
	in both cases.
	We can subdivide the polygon into $O(n/s)$ subpolygons of complexity
	$O(s)$.	
	Therefore, the theorem holds.
\end{proof}

\section{Applications}
We first introduce other subdivision methods frequently used for $s$-workspace algorithms
and provide comparison for our balanced subdivision method with them. Then we will present $s$-workspace
algorithms that improve the previously best known results for three problems without increasing the size of
the workspace.

\label{sec:application}
\subsection{Comparison with Other Subdivision Methods}
There are several subdivision methods which are used for computing the
shortest path between two points in the context of time-space
trade-offs.  Asano et al.~\cite{asano_memory-constrained_2013}
presented a subdivision method that subdivides a simple $n$-gon into
$O(s)$ subpolygons of complexity $O(n/s)$ using $O(s)$ chords. 
They
showed that the shortest path between any two points in the
  polygon can be computed in $O(n^2/s)$ time using $O(s)$ words of
workspace.  However, their algorithm takes $O(n^2)$ time to compute
the subdivision, which dominates the overall running time.
In fact, in the paper they asked whether a subdivision for computing shortest paths 
can be computed more efficiently using $O(s)$ words of workspace.

Instead of answering this question directly,
Har-Peled~\cite{har-peled_shortest_2015} presented a way to subdivide
a simple $n$-gon into $O(n/s)$ subpolygons of complexity $O(s)$. 
The number of segments defining this subdivision
can be strictly larger than $O(s)$, for $s=\omega(\sqrt{n})$,
and therefore the whole subdivision may not be stored in the $O(s)$ words of workspace.
Instead, they gave a procedure to find the
subpolygon of the subdivision containing a query point in $O(n+s\log s\log^4 (n/s))$
expected time without maintaining the subdivision explicitly.
They showed that one can find the shortest path
between any two points using this subdivision in a way similar to the
algorithm by Asano et al. in $O(n^2/s+(n/s)T(n,s))$ time, where $T(n,s)$ is the time for computing
the subpolygon of the subdivision containing a query point.
Therefore, the running time is $O(n^2/s+s\log s\log^4 (n/s))$.

The balanced subdivision that we propose 
can replace the subdivision methods in the algorithms by Asano et al. and Har-Peled for
computing the shortest path between any two points.
Moreover, our subdivision method has two advantages
compared to the subdivision methods by Asano et al. and Har-Peled: 
(1) the subdivision can be computed faster than the one by Asano et al., and
(2) we can keep the whole subdivision in the workspace unlike the one by
Har-Peled. By using our balanced subdivision, we can improve 
the running times of trade-offs that use a subprocedure of 
computing the shortest path between two points.
Moreover, we can solve other application problems efficiently using $O(s)$ words of
workspace.
An example is to compute the shortest path between 
a query point and a fixed point after
preprocessing the input polygon for the fixed point. See Lemma~\ref{lem:path-faster}.

\subsection{Time-space Trade-offs Based on the Balanced Subdivision Method} 
By using our balanced subdivision method, we improve 
the previously best known running times for the following three problems without increasing 
the size of the workspace.
 \paragraph{Computing the  shortest path between two points.}  
 Given any two points $p$ and $q$ in $P$,
 we can report the edges of
 the shortest path $\pi(p,q)$ in order in $O(n^2/s)$ deterministic
 time using $O(s)$ words of workspace.  This improves the
 $s$-workspace randomized algorithm by
 Har-Peled~\cite{har-peled_shortest_2015} which takes
 $O(n^2/s + n\log s \log^4(n/s))$ expected time.
 
 We can compute the
 shortest path between two query points using our balanced subdivision
 as follows.  For $s \geq \sqrt{n}$, we have the subdivision
 consisting of $O(n/s)$ subpolygons of complexity $O(s)$.  Thus we use the
 algorithm by Har-Peled~\cite{har-peled_shortest_2015} described in
 the following lemma.  Har-Peled presented an algorithm that for a given
 query point $q$ computes the subpolygon of the subdivision
 containing $q$ in $O(n+s\log s\log^4(n/s))$ expected time ~\cite[Lemma
 3.2]{har-peled_shortest_2015}. It is shown in the paper that the shortest path
   between any two points can be computed using the algorithm as stated in the
   following lemma.
 
 \begin{lemma}[Implied by~{\cite[Lemma 4.1 and Theorem 4.3]{har-peled_shortest_2015}}]
   For a subdivision of a simple polygon consisting of $O(n/s)$ subpolygons, each of complexity $O(s)$, 
   if the subpolygon containing a query point can be computed in $T(n)$
   time using $O(s)$ words of workspace, the
   shortest path between any two points can be computed in $O((n/s)(T(n)+n))$ time
   using $O(s)$ words of workspace.
 \end{lemma}
 
 In our case, we can find the subpolygon of the balanced subdivision
 containing a query point in $O(n)$ deterministic time.  Combining
 this result with the lemma, we can compute the shortest path between
 any two points in $O(n^2/s)$ deterministic time.
 
 For $s< \sqrt{n}$, we have the subdivision consisting of $O(s)$ subpolygons
   of complexity $O(n/s)$. Instead of the algorithm by Hal-Peled,
   we use the algorithm by Asano et
   al.~\cite{asano_memory-constrained_2013} to compute the shortest
   path between any two points in the polygon.

 \begin{theorem}
   \label{thm:path}
   Given any two points in a simple polygon with $n$ vertices, we can
   compute the shortest path between them in $O(n^2/s)$ deterministic
   time using $O(s)$ words of workspace.
 \end{theorem}

\paragraph{Computing the shortest path tree from a point.}
The \emph{shortest path tree} rooted at $p$ is defined to be the union
of $\pi(p,v)$ over all vertices $v$ of $P$.  Aronov et
al.~\cite{aronov_time-space} gave an $s$-workspace randomized
algorithm for computing the shortest path tree rooted at a given
point. It uses the algorithm by
Har-Peled~\cite{har-peled_shortest_2015} as a subprocedure
and takes $O((n^2\log n)/s + n\log s\log^5{(n/s)})$ expected
time. If one
uses Theorem~\ref{thm:path} instead of Har-Peled's algorithm, the
running time improves to $O((n^2\log n)/s)$ expected time.  In
Section~\ref{sec:SPT}, we improve this algorithm even further using
properties of our balanced subdivision.

\paragraph{Computing a triangulation of a simple polygon.}
Aronov et al.~\cite{aronov_time-space} presented an $s$-workspace algorithm for 
computing a triangulation of a simple $n$-gon. Their algorithm returns the edges
of a triangulation without repetition in $O(n^2/s+n\log s\log^5(n/s))$ expected time.
It uses the shortest path algorithm by Har-Peled~\cite{har-peled_shortest_2015} as a subprocedure,
which takes $O(O(n^2/s+n\log s\log^4(n/s)))$ expected time.
By replacing this shortest path algorithm with ours in Theorem~\ref{thm:path}, we can obtain 
a triangulation of a simple polygon in $O(n^2/s)$ deterministic time using $O(s)$ words of workspace.

\begin{theorem}
	\label{thm:triangulation-polygon}
	Given a simple polygon with $n$ vertices,
	we can compute a triangulation of the simple polygon 
	by returning the edges of the triangulation without
	repetition in $O(n^2/s)$ deterministic time
	using $O(s)$ words of workspace. 
	\end{theorem}

As mentioned by Aronov et al.~\cite{aronov_time-space}, the 
algorithm can be modified to report the resulting triangles of a
triangulation together with their adjacency information in the same
time if $s\geq \log n$.


\section{Improved Algorithm for Computing the Shortest Path Tree}
\label{sec:SPT}
In this section, we improve the algorithm for computing the shortest path tree
from a given point even further to 
$O(n^2/s+(n^2\log n)/s^c)$ expected time for an arbitrary positive constant $c$.
We use the following lemma given by Aronov et al.~\cite{aronov_time-space}.
\begin{lemma}[{\cite[Lemma~6]{aronov_time-space}}]
\label{lem:shortest-path-constant}
For any point $p$ in a simple $n$-gon,
we can compute the shortest path tree rooted at a point in the polygon in $O(n^2\log n)$
expected time using $O(1)$ words of workspace.
\end{lemma}

We apply two different algorithms depending on the size of the workspace: $s=O(\sqrt{n})$ or
$s=\Omega(\sqrt{n})$. We consider the case of $s=O(\sqrt{n})$ first. For the case of
$s=\Omega(\sqrt{n})$, we can store all edges of each subpolygon in the workspace.
\subsection{Case of  \texorpdfstring{$s = O(\sqrt{n})$}{s=O(sqrt{n})}}
\label{sec:small}
Given a point $p\in P$, we want to report all edges of the shortest path tree rooted at $p$.
Recall that there are $O(s)$ extensions on the balanced subdivision in this case. 
We call an edge of a path a \emph{w-edge} if it crosses an extension.
For every extension $a_1a_2$ of the balanced subdivision, we first compute the w-edges of 
$\pi(p,a_1)$ and $\pi(p,a_2)$ in $O(n^2/s^2)$ time in Section~\ref{sec:wall-edges}.
We show that the total number of the w-edges for the two paths is $O(s)$  for every extensions.
These w-edges allow us to compute the shortest path $\pi(p,q)$ for any point $q$ of $P$ 
in $O(n^2/s^2)$ time.

Then we decompose $P$ into subpolygons associated with vertices 
in Section~\ref{sec:all-edges}. For each subpolygon,
we compute the shortest path tree rooted at its associated vertex 
inside the subpolygon recursively.
If a subpolygon satisfies one of \textit{the stopping criteria} (to be defined later), we 
stop the recursion but proceed further to complete the shortest path
tree inside the subpolygon if necessary. 
Because of the space constraint, we restrict the depth of the recurrence to be a constant.

\subsubsection{Computing w-edges}
\label{sec:wall-edges}
We compute all w-edges of the shortest paths between $p$ and 
the endpoints of the extensions.
The following lemma implies that there are $O(s)$ w-edges of the shortest paths. 
For any three points $x, y$ and $z$ in $P$, 
we call a point $x'$ the \emph{junction} of $\pi(x,y)$ and $\pi(x,z)$
if $\pi(x,x')$ is the maximal common path of $\pi(x,y)$ and $\pi(x,z)$.

\begin{lemma}
  \label{lem:wall-edge}
  For an extension $a_1a_2$, there is at most one w-edge of $\pi(p,a_i)$
  for $i=1,2$ which is not a w-edge of $\pi(p,b)$ or $\pi(p,b')$ for any
  other extension $bb'$ crossed by $\pi(p,a_i)$.
\end{lemma}
\begin{proof}
  Let $b_1$ and $b_2$ be the endpoints of the first extension that we
  encounter during the traversal of $\pi(p,a_1)$ from $a_1$ towards
  $p$.  See Figure~\ref{fig:path-tree}(a).  Let $v$ be the junction
    closer to $a_1$ between the junction of $\pi(p,a_1)$ and $\pi(p,b_1)$ and
  the junction of $\pi(p,a_1)$ and $\pi(p,b_2)$.
	
  Note that $\pi(p,a_1)$ is the concatenation of $\pi(p,v)$ and
  $\pi(v,a_1)$.  The vertices of $\pi(v,a_1)$ other than $v$ lie in
  the subpolygon incident to $a_1a_2$ and $b_1b_2$.  Thus every edge of
  $\pi(v,a_1)$ not incident to $v$ is contained in this subpolygon, and does
  not cross any extension.  Therefore, the w-edge of $\pi(p,a_1)$ which is
  not a w-edge of $\pi(p,b)$ or $\pi(p,b')$ for any extension $bb'$ crossed
  by $\pi(p,a_1)$ is unique: the edge of $\pi(v,a_1)$ incident to $v$.
\end{proof}

\begin{figure}
  \begin{center}
    \includegraphics[width=0.75\textwidth]{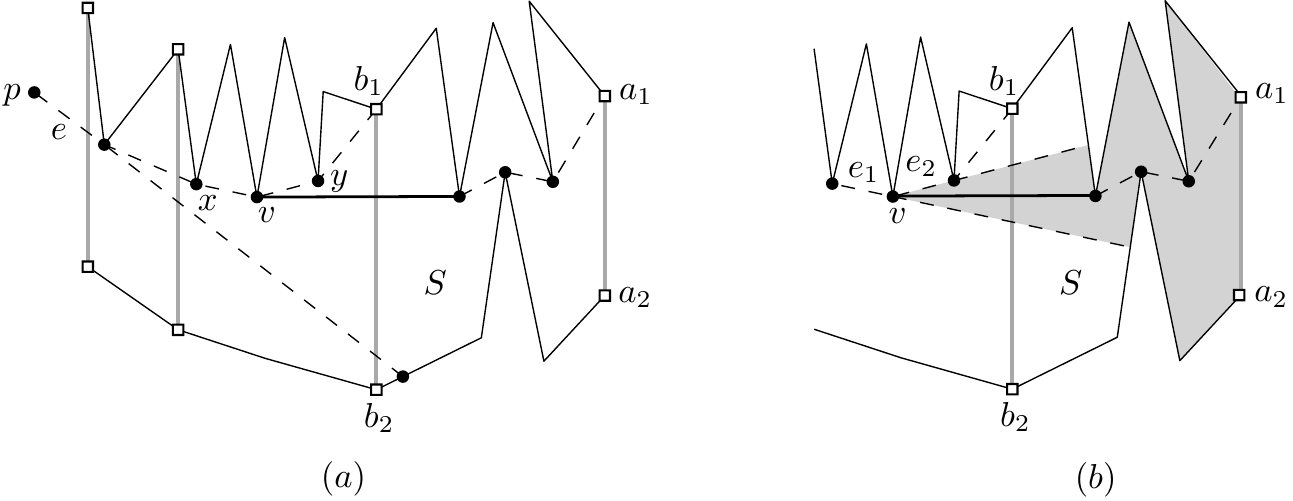}
    \caption{\label{fig:path-tree}\small (a) We compute the junction
      $v$ of $\pi(p,a_1)$ and $\pi(p,b_1)$ by applying binary search
      on the w-edges of $\pi(p,b_1)$.  (b) We extend $e_1$ and $e_2$
      towards $b_1$. The gray region contains the edge of $\pi(v,a_1)$
      incident to $v$ and has complexity of $O(n/s)$.}
  \end{center}
\end{figure}

We consider the extensions one by one in a specific order and compute such
w-edges one by one.  To decide the order for considering the extensions, we
define a \emph{w-tree} $T$ as follows.
Each node $\alpha$ of $T$ corresponds to an extension $d(\alpha)$ of the balanced subdivision of $P$,
except for the root. Also, each extension of the balanced subdivision of $P$ corresponds to a node of $T$. 
The root of $T$ corresponds to $p$ and has
children each of which corresponds to an extension incident to the subpolygon
containing $p$.  A non-root node $\beta$ of $T$ is the parent of a
node $\alpha$ if and only if $d(\beta)$ is the first extension that we
encounter during the traversal of $\pi(p,a_1)$ from $a_1$ for an
endpoint $a_1$ of $d(\alpha)$.  We can compute $T$ in $O(n)$ time.

\begin{lemma}
  The w-tree can be built in $O(n)$ time using $O(s)$ words of
  workspace.
\end{lemma}
\begin{proof}
  We create the root, and its children by traversing the boundary of
  the subpolygon $S_p$ containing $p$.  Then for each subpolygon $S$ incident to
  $S_p$, we traverse its boundary. Let $\alpha$ be the node of the
  tree corresponding to the extension incident to both $S$ and $S_p$.  We
  create nodes for the extensions incident to $S$ other than $d(\alpha)$ as
  children of $\alpha$.  We repeat this until we visit every extension of the
  balanced subdivision.  In this way, we traverse the boundary of each
  subpolygon exactly once, thus the total running time is $O(n)$.
\end{proof}

After constructing $T$, we apply depth-first search on $T$. Let $\mathcal{D}$ be an empty set.
When we visit a node $\alpha$ of $T$, we compute the w-edges of $\pi(p,a_1)$ and $\pi(p,a_2)$ which are not in $\mathcal{D}$ yet, where $a_1a_2=d(\alpha)$, and insert them in $\mathcal{D}$.
Each w-edge in $\mathcal{D}$ has information on the node of $T$ defining it and the subpolygons of the balanced
subdivision containing its endpoints. Due to this information, we can compute the w-edges of $\pi(p,a)$ in order from $a$
in $O(s)$ time for any endpoint $a$ of $d(\alpha)$ and any node $\alpha$ we visited before.
Once the traversal is done, $\mathcal{D}$ contains all w-edges in the shortest paths between $p$ and the 
endpoints of the extensions.

We show how to compute  the w-edge of $\pi(p,a_1)$ which is not in $\mathcal{D}$ yet.
We can compute the $w$-edge of $\pi(p,a_2)$ not in $\mathcal{D}$ yet analogously.
By Lemma~\ref{lem:wall-edge}, there is at most one such edge of $\pi(p,a_1)$.
Moreover, by its proof, it is the edge of $\pi(v,a_1)$ that is incident to $v$. 
Here, $v$ is the junction closer to $a_1$ 
between the junction of $\pi(p,b_1)$ and $\pi(p,a_1)$ 
and the junction of $\pi(p,b_2)$ and $\pi(p,a_1)$,
where $b_1b_2$ is the extension corresponding the parent of $\alpha$.
Thus, to compute the $w$-edge, we first compute the junction of 
 $\pi(p,b_1)$ and $\pi(p,a_1)$ 
and the junction of $\pi(p,b_2)$ and $\pi(p,a_1)$.

\paragraph{Computing junctions.}
We show how to compute the junction $v_1$ of $\pi(p,b_1)$ and $\pi(p,a_1)$ in $O(n^2/s^2)$ time 
for $s=O(\sqrt{n})$. The junction $v_2$ of $\pi(p,b_2)$ and $\pi(p,a_1)$ 
can be computed analogously in the same time. 
Then we choose the one between $v_1$ and $v_2$ that is closer to $a_1$.

To do this, we compute the set of the w-edges in $\mathcal{D}$ appearing on $\pi(p,b_i)$ 
in order from $b_i$ in $O(s)$ time for $i=1,2$.
We denote the set by $\mathcal{D}(b_i)$. Note that the edges in $\mathcal{D}(b_i)$ are 
the w-edges of $\pi(p,b_i)$.
We find two consecutive edges in $\mathcal{D}(b_1)$ containing $v_1$ between them along $\pi(p,b_1)$
by applying binary search on the edges in $\mathcal{D}(b_1)$.

Given any edge $e$ in $\mathcal{D}(b_1)$, we can determine which side of $e$ along $\pi(p,b_1)$ contains $v_1$ 
in $O(n/s)$ time as follows.
We first check whether $e$ is also contained in $\pi(p,b_2)$ in constant time using $\mathcal{D}(b_2)$. If so,
$v_1$ is contained in the side of $e$ along $\pi(p,b_1)$ containing $b_1$. Thus we are done.
Otherwise, we extend $e$ towards $b_1$ until it escapes from $S$, where $S$ is the subpolygon incident to both $a_1a_2$ and $b_1b_2$. See Figure~\ref{fig:path-tree}(a).
Note that the extension crosses $b_1b_2$ since both $\pi(b_1,v_e)$ and $\pi(b_2,v_e)$ are concave
for an endpoint $v_e$ of $e$.
We can compute the point where the extension escapes from $S$ in $O(n/s)$ time by traversing the boundary of $S$ once.
If an endpoint of the extension lies on the part of $\bd S$ between $a_1$ and $b_1$ not containing
$a_2$, $v_1$ lies in the side of $e$ containing $p$ along $\pi(p,b_1)$.
Otherwise, $v_1$ is contained in the other side of $e$.
Therefore, we can find two consecutive w-edges in $\mathcal{D}(b_1)$ containing $v_1$ between them along   $\pi(p,b_1)$ in $O((n/s)\log s)$ time since the size of $\mathcal{D}(b_1)=O(s)$.

The edges of $\pi(p,b_1)$ lying between the two consecutive w-edges are contained in the same subpolygon. 
Let $x$ and $y$ be the endpoints of the two consecutive edges of $\mathcal{D}(b_1)$
contained in the same subpolygon.
Then we compute the edges of $\pi(x,y)$ one by one from $x$ to $y$ inside the subpolygon containing $x$ and $y$.
By Theorem~\ref{thm:path}, we can compute $\pi(x,y)$ in $O(n^2/s^3)$ time since the size
of the subpolygon is $O(n/s)$.
Here, we use extra $O(s)$ words of workspace for computing $\pi(x,y)$.
When the algorithm in Theorem~\ref{thm:path} reports an edge $f$ of $\pi(x,y)$, 
we check which side of $f$ along $\pi(x,y)$ contains $v_1$ in $O(n/s)$ time as we did before.
We repeat this until we find $v_1$.
This takes $O((n/s)^2)$ time since there are $O(n/s)$ edges in $\pi(x,y)$.
Therefore, in total, we can compute the junction $v_1$ 
in $O(s+(n/s)\log s + n^2/s^2)=O(n^2/s^2)$ time since $s=O(\sqrt{n})$.

\paragraph{Computing the edge \texorpdfstring{of $\pi(v,a_1)$}{} incident to the junction \texorpdfstring{$v$}{v}.}
In the following, we compute the edge of $\pi(v,a_1)$ incident to $v$.
We assume that $v$ is the junction of $\pi(v,a_1)$ and $\pi(v,b_1)$.
The case that $v$ is the junction of $\pi(v,a_1)$ and $\pi(v,b_2)$ can be handled
analogously. 
Let $e_1$ and $e_2$ be two edges of $\pi(p,b_1)$ incident to $v$, 
which can be obtained while we compute $v$.
See Figure~\ref{fig:path-tree}(b).
We extend $e_1$ and $e_2$ towards $b_1$ until they escape from $P$ for the first time.
The two extensions and $a_1a_2$ subdivide $P$ into regions. Consider 
the region bounded by the two extensions and $a_1a_2$.
Note that the region can be represented using $O(1)$ words as the boundary
consists of three line segments, one from each of the two extensions and $a_1a_2$, and
two boundary (possibly empty) chains of $P$ connecting the segment of $a_1a_2$ to 
the other segments. 
The number of polygon vertices on the boundary of the region is $O(n/s)$.
Moreover, $\pi(v,a_1)$ is contained in the region.
Thus, the edge of $\pi(v,a_1)$ incident to $v$ inside the region is the edge we want to compute.
We can compute it in $O(n^2/s^3)$ time by applying Theorem~\ref{thm:path} to this region.

In summary, we compute the w-edge of $\pi(p,a_1)$ which has not been computed yet
in $O(n^2/s^2)$ time, assuming that we have done this for every node we have visited so far.
More specifically, computing the junction of $\pi(p,a_1)$ and $\pi(p,b_i)$ takes $O(n^2/s^2)$ time for $i=1,2$, and 
computing the edge incident to each junction takes $O(n^2/s^3)$ time. 
One of the edges is the w-edge that we want to compute.
Since the size of the w-tree is $O(s)$, we can do this for every node in $O(n^2/s)$ time in total.
Thus we have the following lemma.

\begin{lemma}
\label{lem:computing-wall-edges}
Given a point $p$ in a simple polygon with $n$ vertices,
we can compute all w-edges of the shortest paths between $p$ 
and the endpoints of the extensions
in $O(n^2/s)$ time using $O(s)$ words of workspace for $s=O(\sqrt{n})$.
\end{lemma}

Due to the w-edges, we can compute the shortest path $\pi(p,q)$ in
$O(n^2/s^2)$ time for any point $q$ in $P$. Note that $n^2/s^2$ is at least $n$ for $s=O(\sqrt{n})$.

\begin{lemma}
\label{lem:path-faster}
Given a fixed point $p$ in $P$ and a parameter $s=O(\sqrt{n})$, 
we can compute $\pi(p,q)$ in $O(n^2/s^2)$ time for any point $q$ in $P$ using $O(s)$ words of workspace 
after an $O(n^2/s)$-time preprocessing for $P$ and $p$.
\end{lemma}
\begin{proof}
	As a preprocessing, we compute the balanced subdivision of $P$. Then we compute all w-edges of the shortest paths
	between $p$ and the endpoints of the extensions 
	in $O(n^2/s)$ time using Lemma~\ref{lem:computing-wall-edges}.
	
	To compute $\pi(p,q)$, we first find the subpolygon of the balanced subdivision containing $q$ in $O(n)$ time.
	The subpolygon is incident to $O(1)$ extensions due to Corollary~\ref{lem:incident-vertical-line}.
	Consider the nodes in the w-tree corresponding to these extensions. One of the nodes is the parent
	of the others. We find the extension corresponding to the parent and denote it by $a_1a_2$.
	This extension is the first extension crossed by $\pi(p,q)$ we encounter
	during the traversal of $\pi(p,q)$ from $q$.
	
	Then we compute the w-edge $e$ of $\pi(p,q)$ which is not in $\mathcal{D}$
	in $O(n^2/s^2)$ time as we did before, where  $\mathcal{D}$ is the set of all w-edges 
	of the shortest paths between $p$ and the endpoints of the extensions. 
	Let $v$ be the endpoint of $e$ closer to $q$. We report the edges of $\pi(q,v)$ from $q$ one by one
	using the algorithm in Theorem~\ref{thm:path}. Note that $\pi(q,v)$ is contained in a single subpolygon of the balanced subdivision.
	We can report them in $O(n^2/s^3)$ time since the subpolygon has complexity of $O(n/s)$.
	Then we report $e$ as an edge of $\pi(p,q)$. 	
	
	The remaining procedure is to report the edges of $\pi(p,v')$, where $v'$ is the endpoint of $e$ other than $v$. 
	Note that $v'$ lies on $\pi(p,a_1)\cup\pi(p,a_2)$.
	Without loss of generality, we assume that it lies on $\pi(p,a_1)$. We can find all w-edges of $\pi(p,v')$ by computing all w-edges
	of $\pi(p,a_1)$ in $O(s)$ time.
	We consider the w-edges one by one from the one closest to $v'$ to the one farthest from $v'$.
	For two consecutive w-edges $e$ and $e'$ along $\pi(p,v')$, we report the edges of $\pi(p,v)$ lying between $e$ and $e'$.
	This takes $O(n^2/s^3)$ time since all such edges are contained in a single subpolygon of complexity $O(n/s)$.
	Since there are $O(s)$ w-edges, we can report all edges of $\pi(p,q)$ in $O(n^2/s^2)$ time in total.
\end{proof}

\subsubsection{Decomposing the Shortest Path Tree into Smaller Trees}
\label{sec:all-edges}
We subdivide $P$ into subpolygons each of which is associated with a vertex of it in a way different from the one for the balanced subdivision. 
Then inside each such subpolygon, 
we report all edges of the shortest path tree rooted at its associated vertex recursively.
We guarantee that the edges reported in this way are the edges of the shortest path tree rooted at $p$.
We also guarantee that all edges of the shortest path tree rooted at $p$ are reported.
We use a pair $(P',p')$ to denote the problem of reporting the shortest path tree rooted at a point $p'$ inside a simple subpolygon $P'$ of $P$. 
Initially, we are given the problem $(P, p)$.

\paragraph{Structural properties of the decomposition.}
We use the following two steps of the decomposition.
In the first step, we decompose $P$ into a number of subpolygons by the 
shortest path $\pi(p,a)$ for every endpoint $a$ of the extensions.
The boundary of each subpolygon consists of a polygonal chain of $\bd P$ with endpoints $g_1, g_2$ and 
the shortest paths $\pi(v,g_1)$ and $\pi(v,g_2)$, 
where $g_1,g_2$ are endpoints of extensions and $v$ is the junction of $\pi(p,g_1)$ and $\pi(p,g_2)$.
In the second step, we decompose each subpolygon further into smaller subpolygons
by extending the edges of the shortest paths $\pi(v,g_1)$ and $\pi(v,g_2)$ towards $g_1$ and $g_2$,
respectively. See Figure~\ref{fig:path-tree-subproblem}.

Consider a subpolygon $P'$ in the resulting subdivision. Its boundary
consists of a polygonal chain of $\bd P$ and two line segments sharing a common endpoint $p'$.
We can represent $P'$ using $O(1)$ words. Moreover, $P'$ has complexity of $O(n/s)$.
For any point $q$ in $P'$, $\pi(p,q)$ is the concatenation of $\pi(p,p')$ and $\pi(p',q)$. Therefore,
the shortest path rooted at $p'$ inside $P'$ coincides with the shortest path tree rooted at $p$ inside 
$P$ restricted to $P'$.
We can obtain the entire shortest path tree rooted at $p$ inside $P$ by computing it on $(P',p')$ for every subpolygon $P'$ in the resulting subdivision
and its associated vertex $p'$.

We define the orientation of an edge of the shortest path tree using the indices of its endpoints (for example, from a smaller index to a larger index.)
Note that the endpoints of an edge of the shortest path tree are vertices of $P$ labeled from $v_0$ to $v_{n-1}$.
We do not report an edge $e$ of the shortest path if $P'$ contains $e$ on its boundary and lies locally to the right of $e$
for a base problem $(P',p')$. Then every edge is reported exactly once. 
\begin{figure}
	\begin{center}
		\includegraphics[width=0.65\textwidth]{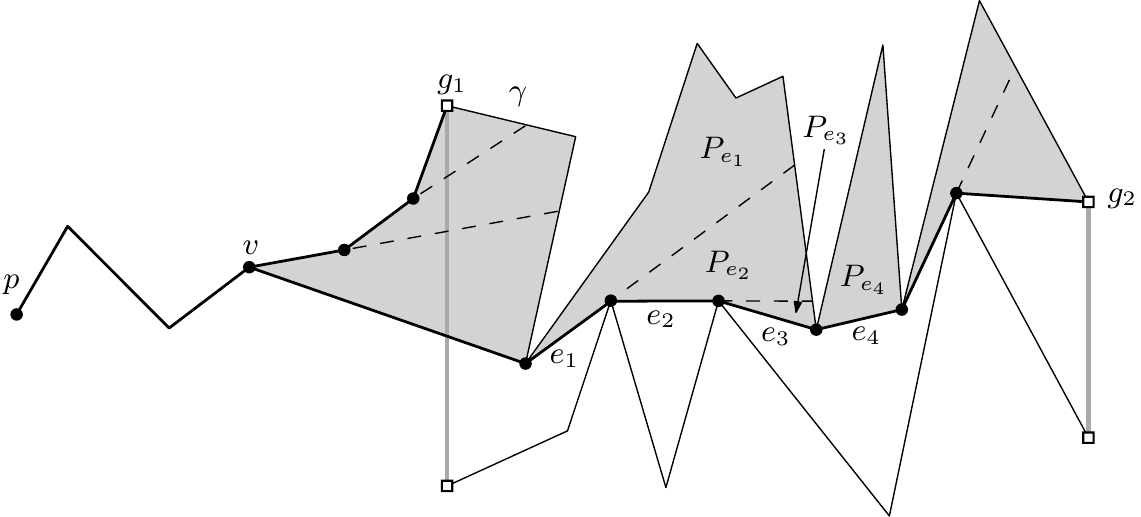}
		\caption{\label{fig:path-tree-subproblem} \small
			Subdivision of the region bounded by $\pi(v,g_1)\cup\pi(v,g_2)$ and the polygonal chain
			of $\bd P$ from $g_1$ to $g_2$ in clockwise order along $\bd P$ by the extensions of 
			edges of $\pi(v,g_1)$ and $\pi(v,g_2)$ towards $g_1$ and $g_2$, respectively. 
			We extend an edge
			of the paths if at least one of its endpoints are not on $\gamma$.}
	\end{center}
\end{figure}

\paragraph{Computing the subpolygons with their associated vertices.}
In the following, we show how to obtain this subdivision.
Recall that the boundary of a subpolygon $S$ in the balanced subdivision consists of extensions and polygonal chains from $\bd P$.
For each maximal polygonal chain $\gamma$ of $\bd S$ containing no endpoint 
of extensions in its interior, we do the followings.
Let $g_1$ and $g_2$ be the endpoints of $\gamma$.
We compute the junction $v$ of $\pi(p,g_1)$ and $\pi(p,g_2)$ in $O(n^2/s^2)$ time
as we did in Section~\ref{sec:wall-edges}.

Consider the region (subpolygon) of $P$ bounded by $\pi(v,g_1)$, $\pi(v,g_2)$ and $\gamma$. 
We compute the edges of $\pi(p,g_i)$ lying between $v$ and $g_i$ in order in 
$O(n^2/s^2)$ time using Lemma~\ref{lem:path-faster} for $i=1,2$.
Clearly, these edges are the edges of $\pi(v,g_i)$.
Whenever we compute an edge $e$ of $\pi(v,g_i)$, we check whether the endpoints of $e$
are on $\gamma$ or not, and obtain a subproblem $(P_e,v_e)$ as follows. 
Let $v_e$ be the endpoint closer to $v$.
See Figure~\ref{fig:path-tree-subproblem} for an illustration. 

\begin{itemize}
	\item Both endpoints are on $\gamma$: 
	$P_e$ is the subpolygon bounded by $e$ and a part of $\gamma$ 
	connecting the two endpoints of $e$. 
	(See $e_4$ in Figure~\ref{fig:path-tree-subproblem}.)
	\item Exactly one of the endpoints are on $\gamma$: 
	If $v_e$ is not on $\gamma$, we extend the edge incident to $v_e$ other than $e$ towards $g_i$ until it hits $\gamma$ 
	in $O(n/s)$ time. (See $e_3$ in Figure~\ref{fig:path-tree-subproblem}.)
	If $v_e$ is on $\gamma$, we extend $e$ towards $g_i$ until it hits $\gamma$ 
	in $O(n/s)$ time.
	(See $e_1$ in Figure~\ref{fig:path-tree-subproblem}.)
	Let $P_e$ is the subpolygon bounded by the extension (including $e$) 
	and the part of $\gamma$ connecting
	the endpoint of $e$ and the endpoint of the extension lying on $\gamma$.
	\item No endpoint is on $\gamma$:
	We extend both edges of $\pi(v,g_2)$ incident to $v_e$ towards $g_i$ in $O(n/s)$ time.
	Let $P_e$ be the subpolygon bounded by the two extensions (including $e$) and the part of 
	$\gamma$ connecting the endpoints of the extensions lying on $\gamma$. 
	(See $e_2$ in Figure~\ref{fig:path-tree-subproblem}.)
\end{itemize}

Therefore, we can compute the decomposition of the region of $P$ bounded by 
$\pi(v,g_1)$, $\pi(v,g_2)$ and $\gamma$ in $O(n^2/s^2+nk/s)$ time, 
where $k$ is the number of edges of $\pi(v,g_1)\cup\pi(v,g_2)$ for the junction 
$v$ of $\pi(p,g_1)$ and $\pi(p,g_2)$.
Since there are $O(s)$ such maximal polygonal chains containing no
endpoint of extensions in their interiors and the sum of $k$ over all such maximal polygonal chains
is $O(n)$, the running time for decomposing the problem $(P,p)$ into smaller 
problems is $O(n^2/s)$.
\begin{lemma}
	We can decompose the problem $(P,p)$ into smaller problems 
	in $O(n^2/s)$ time. 
\end{lemma}

We decompose each problem recursively unless the problem satisfies one of 
the three stopping criteria in Definition~\ref{def:stopping}.
Then we solve each base problem directly, that is, we report the edges of the shortest path tree.
But for non-base problems, we do not report any edge of the shortest path tree.

\begin{definition}[Stopping criteria] 
\label{def:stopping}
There are three stopping criteria for $(P',p')$:\\
(1) $P'$ has $O(s)$ vertices, (2) $s\geq \sqrt{|P'|}$,  where $|P'|$ is the complexity of $P'$, 
and (3) the depth of the recurrence is a positive constant $c$.
\end{definition}

When stopping criterion (1) holds, we compute the shortest path tree directly using
the algorithm by Guibas et al.~\cite{guibas_linear_1987}.
When stopping criterion (2) holds, we apply the algorithm described in Section~\ref{sec:large-space-tree}   
that computes the shortest path tree rooted at $p'$ inside $P'$ in $O(|P'|^2/s)$ time for the case that
$s\geq \sqrt{|P'|}$, where $|P'|$ is the complexity of $P'$. 
When stopping criterion (3) holds, we compute the shortest path tree using Lemma~\ref{lem:shortest-path-constant}.


\subsubsection{Analysis of the Recurrence}
\label{sec:faster}

\paragraph{Time complexity.}
Consider the base problems.
All base problems induced by stopping criterion (1) can be handled in $O(n)$ time in total
because the subpolygons corresponding to them are pairwise interior-disjoint.
For the base problems induced by stopping criterion (2), 
the subpolygons corresponding to them are also pairwise interior-disjoint.
The time for handling these problems is the sum of $O(|P'|^2/s)$ over all the problems
$(P',p')$. The running time is $O(n^2/s)$ because we have
$O(\sum |P'|^2/s)=O(n/s\sum |P'|)=O(n^2/s).$

Now we consider the base problems induced by stopping criterion (3).
For depth $c$ of the recurrence, every subpolygon has complexity of $O(n/s^c)$. Moreover, the total complexity of 
all subpolygons at depth $c$ is $O(n)$. By Lemma~\ref{lem:shortest-path-constant},
the expected time for computing the shortest path trees in all subpolygons 
is the sum of $O(|P'|^2 \log n)$ over all subpolygons $P'$ at depth $c$.
Therefore, we can solve all problems at depth $c$ in $O((n^2\log n)/s^c)$ time because we have
$O(\sum_i |P_i|^2\log n)=O(n/s^c \sum_i |P_i| \log n)=O((n^2\log n)/s^c)$.

We analyze the running time for decomposing a problem into smaller problems.
Consider depth $k$ for $1\leq k < c$.
Let $(P_1,p_1),\ldots,(P_t,p_t)$ be the problems at depth $k$.
Note that the sum of $|P_i|$ over all indices from $1$ to $t$ is $O(n)$.
For each $P_i$, we construct the balanced subdivision of $P_i$ in $O(|P_i|^2/s)$ time, 
compute $O(s)$ w-edges of the shortest paths between $p_i$ and the endpoints of the extensions
in $O(|P_i|^2/s^2)$ time,
and decompose the problem into smaller problems in $O(|P_i|^2/s)$ time.
Thus, the decomposition takes $O(\sum_i |P_i|^2/s)=O(n^2/s)$ time for the problems at depth $k$.
Since $c$ is a constant, the decomposition over the $c$ depths takes $O(n^2/s)$ time.

Therefore, the total running time is $O(n^2/s+(n^2\log n)/s^c)$ for an arbitrary
constant $c>0$.

\paragraph{Space complexity.}
To handle each problem $(P',p')$, we maintain the balanced subdivision of $P'$ using $O(s)$ words of workspace.
Until all subproblems of $(P',p')$ for all depths are handled, we keep this balanced subdivision.
However, we do not keep the subdivision for two distinct problems in the same depth at the same time.
Therefore, the total space complexity is $O(cs)$, which is $O(s)$.

\begin{lemma}
	Given a point $p$ in a simple polygon with $n$ vertices,
	we can compute the shortest path tree rooted at $p$ in $O(n^2/s+ (n^2\log n)/s^c)$ expected time 
	using $O(s)$ words of workspace for $s=O(\sqrt{n})$, where $c$ is an arbitrary positive constant. 
\end{lemma}

\subsection{Case of \texorpdfstring{$s=\Omega(\sqrt{n})$}{}}
\label{sec:large-space-tree}
For the case of $s=\Omega(\sqrt{n})$, the balanced subdivision consists of $O(n/s)$ subpolygons
of complexity $O(s)$. The algorithm for this case is similar to the one for the case of $s=O(\sqrt{n})$, 
except that we do not use Theorem~\ref{thm:path} and Lemma~\ref{lem:shortest-path-constant}. 
Instead, we use the fact that we can store all edges of each subpolygon in the workspace.

As we did before, we compute all w-edges of the shortest paths between $p$ and the endpoints of the extensions.
Using them, we decompose $(P,p)$ into a number of subproblems. In this case, we will see that
every subproblem of $(P,p)$ is a base problem due to stopping criterion~(1)
in Definition~\ref{def:stopping}. Then we solve each subproblem directly using the algorithm
by Guibas et al.~\cite{guibas_linear_1987}.

\begin{lemma}
	\label{lem:w-walls-large}
	We can compute all w-edges of the shortest paths between $p$ and the endpoints of the extensions in $O(n)$ time.
\end{lemma}
\begin{proof}
	As we did in Section~\ref{sec:wall-edges}, we apply depth-first search on the w-tree 
	and compute the w-edges one by one.
	When we reach a node $\alpha$ of the w-tree, we compute the w-edges of $\pi(p,a_1)$ and $\pi(p,a_2)$ which are not computed yet,
	where $a_1$ and $a_2$ are the endpoints of the extension $d(\alpha)$. 
	We show how to compute the edges of $\pi(p,a_1)$ only. The 
	case for $\pi(p,a_2)$ can be handled analogously.
	By Lemma~\ref{lem:wall-edge}, there is at most one such w-edge of $\pi(p,a_1)$.
	Moreover, by the proof of the lemma, such an edge is incident to $v$ on $\pi(v,a_1)$,
	where $v$ is the one closer to $a_1$ than the other 
	between the junction of $\pi(p,b_1)$ and $\pi(p,a_1)$ and
	the junction of $\pi(p,b_2)$ and $\pi(p,a_1)$,
	where $b_1b_2$ is the extension corresponding the parent of $\alpha$.
	
	We compute the junction $v_1$ of $\pi(p,b_1)$ and $\pi(p,a_1)$ as follows.
	Consider the endpoints of the w-edges of $\pi(p,b_1)$ sorted along $\pi(p,b_1)$ from $p$.
	We connect them by line segments in this order to form a polygonal chain.
	We denote the resulting polygonal chain by $\mu_1$. Notice that it
	might intersect the boundary of $P$.
	We also do this for $b_2$ and denote the resulting polygonal chain by $\mu_2$.
	We can compute $\mu_1$ and $\mu_2$ in $O(n/s)=O(s)$ time.
	
	Consider the union of the subpolygon of the balanced subdivision incident to both $a_1a_2$ and
	$b_1b_2$, and the region (funnel) bounded by $\mu_1$, $\mu_2$ and $b_1b_2$.
	The complexity of this union is $O(s)$.
	Thus, we can compute the shortest path tree rooted at $p$ restricted in this union using an
	algorithm by Guibas et al.~\cite{guibas_linear_1987}.
	This algorithm computes the shortest path tree rooted at a given point in time linear to the complexity of the input simple polygon
	using linear space.
	We find the maximal subchain of $\mu_1$ which is a part of $\pi(p,a_1)$. One endpoint of the subchain
	is $p$. Let $v'$ be the other endpoint.
	We find two consecutive w-edges $e$ and $e'$ of $\pi(p,b_1)$ containing $v'$ between them along $\pi(p,b_1)$.
	Then they also contain the junction $v_1$ between them along $\pi(p,b_1)$.
	
	Let $x$ and $y$ be the endpoints of $e$ and $e'$, respectively, that are contained in the 
	same subpolygon. We compute $\pi(x,y)$ one by one using the algorithm by Guibas et al.
	We compute the extensions of the edges of $\pi(x,y)$ towards the subpolygon
	containing $a_1a_2$ and $b_1b_2$ on its boundary in $O(s)$ time.
	Then we can decide which vertex of $\pi(x,y)$ is the junction $v_1$.
	This takes $O(s)$ time in total.
	
	Moreover, while we compute $v_1$, we can obtain the edge of $\pi(v_1,a_1)$ incident to $v_1$.
	Thus, we can obtain the w-edge of $\pi(p,a_1)$ which is not computed yet in $O(s)$ time.
	Since there are $O(n/s)$ nodes in the w-tree, we can compute all w-edges of the shortest paths between $p$ and 
	the endpoints of the extensions in $O(n)$ time in total.
\end{proof}

We decompose the problem $(P,p)$ into smaller problems in $O(n^2/s)$ time in a way
similar to the one in Section~\ref{sec:all-edges}. 
\begin{lemma}
	We can decompose the problem $(P,p)$ into smaller problems defined in
	Section~\ref{sec:all-edges} in $O(n^2/s)$ time.
\end{lemma}
\begin{proof}
	Recall that the boundary of a subpolygon $S$ in the balanced subdivision consists of extensions 
	and polygonal chains from $\bd P$.
	For each maximal polygonal chain $\eta$ of $\bd S$ containing no endpoint 
	of extensions in its interior, we do the followings.
	Let $g_1$ and $g_2$ be the endpoints of $\eta$.
	We compute the junction $v$ of $\pi(p,g_1)$ and $\pi(p,g_2)$ in $O(s)$ time
	as we showed in the proof of Lemma~\ref{lem:w-walls-large}.
	
	Then we compute the w-edges of $\pi(v,g_1)$ and $\pi(v,g_2)$ in $O(n/s)=O(s)$ time.
	We are to compute the first point hit by the extension (ray) of each w-edge
	towards $\eta$.
	To do this, we connect the w-edges by line segments 
	to form a polygonal chain $\mu$ as we
	did in Lemma~\ref{lem:w-walls-large}.
	We compute the union of $\mu$ and the part of $\eta$ connecting the two endpoints
	of $\mu$, which is a simple polygon.
	Then we apply the shortest path tree
	algorithm by Guibas et al.~\cite{guibas_linear_1987}.
	This takes $O(s)$ time since there are $O(s)$ such w-edges and $\eta$ has complexity of $O(s)$.
	
	For the edges of $\pi(v,g_1)$ and $\pi(v,g_2)$ lying between two consecutive w-edges, we observe that they
	are contained in the same subpolygon of the balanced subdivision. Thus we can compute such edges and
	extend them towards $\eta$ by applying the algorithm by Guibas et al.
	For a pair of consecutive w-edges, we can do this in $O(s)$ time.
	Since there are $O(n/s)$ such pairs, this takes $O(n)$ time for each maximal polygonal chain $\eta$.
	There are $O(n/s)$ maximal polygonal chains $\eta$, and thus the total running time is $O(n^2/s)$.
\end{proof}

Note that the boundary of each subpolygon $P'$ consists of two line segments
and a part of $\bd P$ containing no endpoint of extensions in its interior. Thus, the complexity of $P'$ is $O(s)$.
This means that all subproblems of $(P,p)$ are base problems due to stopping criterion (1)
in Definition~\ref{def:stopping}. We can solve all base problems in $O(n)$ time in total.
Therefore, we can compute the shortest path tree in $O(n^2/s)$ deterministic time in total.
\begin{lemma}
	Given a point $p$ in a simple polygon with $n$ vertices,
	we can compute the shortest path tree rooted at $p$ in $O(n^2/s)$ deterministic time
	using $O(s)$ words of workspace for $s = \Omega(\sqrt{n})$.
\end{lemma}

Combining the algorithm for case of $s=O(\sqrt{n})$ in Section~\ref{sec:small} with the lemma above, 
we have the following theorem.

\begin{theorem}\label{thm}
Given a point $p$ in a simple polygon with $n$ vertices,
we can compute the shortest path tree rooted at $p$ in $O(n^2/s+(n^2\log n)/s^c)$ expected time
using $O(s)$ words of workspace for an arbitrary positive constant $c$.
\end{theorem}

Here, the size of the workspace is $O(cs)$. Thus, by changing the roles of $c$ and $s$, we can
achieve another $s$-workspace algorithm. In specific, 
by setting $c$ to the size of workspace and $s$ to $2$, we have the following theorem. 

\begin{theorem}
	Given a point $p$ in a simple polygon with $n$ vertices,
	we can compute the shortest path tree rooted at $p$ in $O((n^2\log n)/2^s)$ expected time
	using $O(s)$ words of workspace for $s\leq \log\log n$.
\end{theorem} 

\section{Conclusion}
We present an $s$-workspace algorithm
for computing a balanced subdivision of a simple polygon consisting of 
$O(\min\{n/s,s\})$ subpolygons of complexity $O(\max\{n/s,s\})$. 
This subdivision can be computed more efficiently than other subdivisions suggested 
in the context of time-space trade-offs, and therefore can be used for solving several
fundamental problems in a simple polygon more efficiently.
Since our subdivision method keeps all extensions of the balanced subdivision in the workspace,
it has a few other application problems, including
the problem for answering a single-source shortest path query.
We also believe that we can preprocess a simple polygon and 
maintain a data structure of size $O(s)$ so that for any
two points $x$ and $y$ in a simple polygon $\pi(x,y)$ can be computed in $o(n^2/s)$ time with $O(s)$ words of workspace
by combining the ideas from Guibas and Hershberger~\cite{guibas_optimal_1989} with
our subdivision method. We leave this as a future work.

\bibliographystyle{abbrv} \bibliography{paper}{}
\end{document}